\documentclass{new_tlp}
\usepackage{mathptmx}
\usepackage{enumitem}
\DeclareMathAlphabet{\mathcal}{OMS}{cmsy}{m}{n}
\newcommand\bcmdtab{\noindent\bgroup\tabcolsep=0pt%
  \begin{tabular}{@{}p{10pc}@{}p{20pc}@{}}}
\newcommand\ecmdtab{\end{tabular}\egroup}

  \title[Founded (Auto)Epistemic Equilibrium Logic Satisfies Epistemic Splitting]
        {Founded (Auto)Epistemic Equilibrium Logic Satisfies Epistemic Splitting\thanks{This work has been partially supported by the Centre International de Math\'{e}matiques et d'Informatique de Toulouse (CIMI) through contract ANR-11-LABEX-0040-CIMI within the programme ANR-11-IDEX-0002-02 and the Alexander von Humboldt Foundation.}}

  \author[Jorge Fandinno]
         {JORGE FANDINNO\\
         IRIT, University of Toulouse, CNRS, France\\
			\email{ jorge.fandinno@irit.fr}
			\\
			Universit\"{a}t Potsdam, Germany\\
            \email{fandinno@uni-potsdam.de}}

\jdate{April 2019}
\pubyear{2019}
\pagerange{\pageref{firstpage}--\pageref{lastpage}}

\usepackage{amsmath}
\usepackage{amssymb}
\usepackage{xspace}
\usepackage{enumitem}
\usepackage{IEEEtrantools}

\DeclareMathOperator{\bL}{\mathbf{K}}

\DeclareMathOperator{\bK}{\mathbf{K}}
\DeclareMathOperator{\bM}{\mathbf{M}}

\newcommand\sfequilibrium{\mbox{S5-equilibrium}\xspace}

\newcommand\kdequilibrium{\mbox{equilibrium}\xspace}

\newcommand\sfmodel{\mbox{S5-model}\xspace}

\newcommand\kdmodel{\mbox{D45-model}\xspace}

\newcommand\wv{\mathbb{W}}

\newcommand\kdint[2]{({#2, #1})}

\def\MEQ{{\rm {EQB}}}
\def\WMEQ{{\rm {WEQB}}}
\def\WEQ{{\rm {WEQ}}}
\def\EQ{{\rm {EQ}}}

\def\Atoms{\mathit{Atoms}}
\def\Head{\mathit{Head}}
\def\Body{\mathit{Body}}
\def\Bodym{\mathit{Body}_{\mathit{sub}}}
\def\Bodyr{\mathit{Body}_{\mathit{obj}}}
\def\Bodyp{\mathit{Body}^{+}}
\def\Bodyn{\mathit{Body}^{-}}
\def\Bodyrp{\mathit{Body}^+_{\mathit{obj}}}
\def\Bodymp{\mathit{Body}^+_{\mathit{sub}}}
\def\Bodyrn{\mathit{Body}^-_{\mathit{obj}}}
\def\Bodymn{\mathit{Body}^-_{\mathit{sub}}}

\def\modelsm{\models_{\text{\rm \tiny m}}}
\def\modelshtm{\models_{\text{\rm \tiny htm}}}

\def\modelsm{\models_{\text{\rm \tiny m}}}

\def\SM{\text{\rm SM}}
\def\CL{\text{\rm CL}}

\def\modelsfeel{\models_{\text{\rm \tiny FEEL}}}
\def\equivfeel{\equiv_{\text{\rm \tiny FEEL}}}

\def\At{\text{\em At}}
\def\at{\text{\em At}}
\def\Not{\text{\tt not} \ }
\def\K{\mathbf{K}\, }
\def\M{\mathbf{M}\, }

\def\sS{\mathcal{S}}

\def\At{\text{\em At}}

\def\Litobj{\text{\em Lit}^{\mathit{obj}}}

\def\Not{\text{\tt not} \ }
\def\K{\mathbf{K}\, }
\def\M{\mathbf{M}\, }

\def\Head{\mathit{Head}}
\def\Body{\mathit{Body}}
\def\Bodyp{\mathit{Body}^{+}}
\def\Bodyn{\mathit{Body}^{-}}
\def\Atoms{\mathit{Atoms}}

\def\SM{\text{\rm SM}}
\usepackage[comments=off,changes=off]{krudces}

\ifdefined\WITHREVIEWS %%%%%%%%%%%%%%%%
\definecolor{darkred}{rgb}{0.5,0.0,0.1}
\newcommand{\review}[2]{{\color{darkred} #1}\marginpar{\footnotesize {\color{darkred} #2}}}
\newcommand{\reviewcol}{\color{darkred}}
\newcommand{\rev}[1]{{\color{blue} #1}}
\else
\definecolor{darkred}{rgb}{0.0,0.0,0.0}
\newcommand{\review}[2]{#1}
\newcommand{\reviewcol}{}
\fi
\begin{document}

%%%%%%%%%%%%%%%%%%%%%%%%%%%%%%%%%%%%%%%%%%%%%%%%%%%%%%%%%%%%%%%%%%%
\ifdefined\WITHREVIEWS
%!TEX root = main.tex

\newpage
\newpage
\section*{Summary of changes}

\pagestyle{empty}

I wish to thank the reviewers for their thorough evaluation and their useful comments and suggestions.
Each comment is answered below using the code R$x$.$yy$ for comment number $yy$ from reviewer $x$.
To facilitate the reviewing process, each change is colored in red and makes a marginal reference to the reviewer's comment code R$x$.$yy$.

\subsection*{\bf Reviewer 1}
\begin{enumerate}[label=R1.\arabic*]
\item\label{rev1.1} \rev{The author appear to consider having a unique world view
as a desirable property of a semantics. Other authors however,
for instance Shen and Eiter in a paper cited here, show that on 
significant examples having more than one world view makes sense.
So, a short discussion about this position would be in order.}

The author does not suggest that having a unique world view
is a desirable property,
instead just shown that there exists a class of programs (those epistemically stratified) which has at most one world view.
This does not mean that this is a desirable property in general, but it is hard to justify why an epistemically stratified program could have more than one world view.
% Note that, the author is among the authors of~\cite{CabalarFF2019splitting}, in which it is shown how any semantics satisfying Epistemic Splitting can use the usual ASP representation of action domains to obtain conformant plans in a way that there is a one-to-one correspondence between the these plans and the world views of the program.
% This approach requires several world views, one for each possible conformant plan and follows the usual generate-define-test methodology, also showing that having several world views is useful.
% A different approach to represent conformant plans that also makes use of several world views in Epistemic Logic programs was first discussed in~\cite{Kahl15}.

\item\label{rev1.2} \rev{Overall, the main result of the paper is that FAEEL satisfies the epistemic splitting property;
this is an important property, so the result is valuable.
The paper should be revised in order to better outline which are the new contributions,
as in the present form the reader may be confused on that.}

Done

\item\label{rev1.3} \rev{There is a confusion with names, as FAEEL and FEEL and EEL are very similar, and
are often used together in the same sentence. So, a reader risks to go lost.}

Now they are written as Founded Autoepistemic EEL, Founded EEL and EEL when occur in the same sentence.

\item\label{rev1.4} \rev{It should be said explicitly that you represent a set of interpretations in '[...]'}

This is mentioned in the introduction, page 2, as the first footnote.

\item\label{rev1.5} \rev{pg. 4 third line after Example 1 '(as usual...)' -> '(that as usual...)'}

Done

\item\label{rev1.6} \rev{pg. 4 Definition 1 Atoms(...), although intuitive, has not been defined}

The definition has been added in the background sections.

\item\label{rev1.7} \rev{pg 4 line before Definition 2 'obtaining thys' -> 'thus obtaining'}

Done

\item\label{rev1.8} \rev{pg. 4 Sect. 2.2 line 2 'belief view to the whole set'->'belief view the whole set'}

Done

\item\label{rev1.9} \rev{pg 5 Example 2 'none of the two'??? two what?}

Corrected

\item\label{rev1.10} \rev{pg. 5 Example 3 line 6 'world vies'->'world view'}

Done

\item\label{rev1.11} \rev{pg. 6 Definition 3: 'we define an' -> 'we define a'
Two lines after, 'was replaced by replaced' -> 'was replaced by' and
'in other words'-> 'intuitively'.}

Done

\item\label{rev1.12} \rev{pg. 7 fifth bullet: how can $W^t$, i.e., a set of sets, imply something or not?
This is not defined.}

A set of sets is a belief view as defined in section~\ref{sec:g91}.
This section also defines entailment for this class of belief views.
A footnote has been added here recalling this fact.

\rev{Below: 'A belief model just captures...' please explain.}

Done

\item\label{rev1.13} \rev{pg. 8 'Technically, this make it simpler' it would be 'makes', but what
do you refer to? Please rephrase the sentence.}

\item\label{rev1.14} \rev{pg. 8 'for all HT-interpretation' -> 'for all HT-interpretations'}

Done

\item\label{rev1.15} \rev{Proposition 1 again the problem of $W^t$ implies}

See~\ref{rev1.12}

\item\label{rev1.16} \rev{pg. 9 'stablish' -> 'establish'}

Done

\item\label{rev1.17} \rev{pg. 9 Property 2: 'and $\Pi$ has no' -> 'or $\Pi$ has no'}

This changes is not correct. It has been rewritten as
``either~$\Pi$ has a unique $\sS$-world view $\wv=\SM[\Pi] \neq \emptyset$ or both $\SM[\Pi]=\emptyset$ and $\Pi$ has no $\sS$-world view at all.'' for further clarification.

\item\label{rev1.18} \rev{pg. 9 Example 7: 'Note that using Theorem 4'-> 'Note that, by using Theorem 4,'}

Done

\item\label{rev1.19} \rev{pg. 10 'Furthermore, the following ...' please reformulate this awkward sentence}

Done

\item\label{rev1.20} \rev{pg. 11 'for any semantics that satisfy' -> 'for any semantics that satisfies'}

Done

\item\label{rev1.21} \rev{pg. 11 'coincides with the for the class' ???}

Corrected

\item\label{rev1.22} \rev{pg. 11 Theorem 6 'that satisfies' -> 'that satisfy'}

Done

\item\label{rev1.23} \rev{pg 13 Conclusions 'a desirable property for epistemic
logic programs that. among previous semantics' ???}

Corrected

\end{enumerate}

%%%%%%%%%%%%%%%%%%%%%%%%%%%%%%%%%%%%%%%%%%%%%%%%%%%%%%%%%%%%%%%%%%%%%%%%
\subsection*{\bf Reviewer 2}
\begin{enumerate}[label=R2.\arabic*]
\item\label{rev2.1} \rev{The paper is well written, even though it requires careful reading to take care of simple typos.}

Thanks, the typos have been revised.

\item\label{rev2.2} \rev{The paper (justifiably) talks about similarities between splitting of logic programs by Lifschitz and Turner and epistemic splitting.
But, since the formalism here is that of autoepistemic logic, one can also draw parallels with paper by Gelfond
"On stratified autoepistemic theories" (1987) in which negation as failure was first interpreted in terms of autoepistemic logic
and "On consistency and completeness of autoepistemic theories" by Gelfond and Przymusinska, 
1991 Fundamenta Informaticae 16(1):59-92 which contains a variant of splitting set theorem for autoepistemic logic.}

Thank you for pointing out these works. Appropriate references have been made and Corollary~\ref{cor:stratified.Moore}, which is complementary to Theorem 4 in Gelfond1987,  has been added.

\end{enumerate}

\newpage
%%%%%%%%%%%%%%%%%%%%%%%%%%%%%%%%%%%%%%%%%%%%%%%%%%%%%%%%%%%%%%%%%%%%%%%%
\subsection*{\bf Reviewer 3}
\begin{enumerate}[label=R2.\arabic*]
\item\label{rev3.1} \rev{as author said in the conclusion, there is not a formal comparison of the FAEEL semantics with other existing semantics of the epistemic specification so far, which makes the contribution of this paper incomplete and feeblish.}

Section~\ref{sec:related} has now been renamed as related work and all the semantics are compared in terms of the properties they satisfy.
Recall also that the paper formally prove that for epistemically tight programs FAEEL and G91 coincide.

\item\label{rev3.2} \rev{In the section 2, the satisfaction of $\neg$ operator is not given (Maybe it is easy to understand for most person in this field, but that should appear in a self-containd paper).}

Note that, in the background, $\neg$ is not considered as a connective. Instead, $\neg\varphi$ is defined as an abbreviation for the formula $\varphi \to \bot$.
The satisfaction of $\neg\varphi$ can be obtained by replacing it by its definition.
A proposition stating this has been added.
\end{enumerate}

\newpage
\pagestyle{plain}
\setcounter{page}{1}
\fi
%%%%%%%%%%%%%%%%%%%%%%%%%%%%%%%%%%%%%%%%%%%%%%%%%%%%%%%%%%%%%%%%%%%

\label{firstpage}

\maketitle

\begin{abstract}
In a recent line of research, two familiar concepts from logic programming semantics (unfounded sets and splitting) were extrapolated to the case of epistemic logic programs.
The property of \emph{epistemic splitting} provides a natural and modular way to understand  programs without epistemic cycles but, surprisingly, was only fulfilled by Gelfond's original semantics (G91), among the many proposals in the literature.
On the other hand, G91 may suffer from a kind of self-supported, unfounded derivations when epistemic cycles come into play.
Recently, the absence of these derivations was also formalised as a property of epistemic semantics called \emph{foundedness}.
Moreover, a first semantics proved to satisfy foundedness was also proposed, the so-called \emph{Founded Autoepistemic Equilibrium Logic} (FAEEL).
In this paper, we prove that FAEEL also satisfies the epistemic splitting property something that, together with foundedness, was not fulfilled by any other approach up to date.
To prove this result, we provide an alternative characterisation of FAEEL as a combination of G91 with a simpler logic we called \emph{Founded Epistemic Equilibrium Logic}~(FEEL), 
% This leaves us only with the task to prove that FEEL satisfies epistemic splitting
% and we obtain thaf FAEEL satisfies epistemic splitting as corollary of this result plus the known fact that Gelfond's semantics does.
which is somehow an extrapolation of the stable model semantics to the modal logic S5.
\end{abstract}

\begin{keywords}
Answer Set Programming,
Epistemic Specifications,
Epistemic Logic Programs
\end{keywords}

\section{Introduction}

The language of \emph{epistemic specifications}~\cite{Gelfond91} or \emph{epistemic logic programs} extends disjunctive logic programs, under the \emph{stable model}~\cite{GL88} semantics, with modal constructs called \emph{subjective literals}.
These constructs allow to check whether a regular or \emph{objective literal} $l$ is true in \emph{every} stable model (written $\bK l$) or in \emph{some} stable model (written $\bM l$) of the program.
For instance, the rule:
\begin{eqnarray}
a \leftarrow \neg\!\bK b \label{f:loop1}
\end{eqnarray}
means that $a$ should be derived whenever we cannot prove that all the stable models contain~$b$.
%
% Subjective literals have been incorporated as an extension of the Answer Set Programming (ASP) paradigm~\cite{MT99,Nie99} in different solvers and implementations -- see~\cite{LK18} for a recent survey.
%
The definition of a ``satisfactory'' semantics for epistemic specifications has proven to be a \mbox{non-trivial} enterprise with a long list of alternative semantics~\cite{Gelfond91,WangZ05,Truszczynski11,Gelfond11,CerroHS15,Kahl15,ShenE17,CabalarFF2019faeel}.
The main difficulty arises because subjective literals query the set of stable models but, at the same time, occur in rules that determine those stable models.
As an example, the program consisting of:
\begin{eqnarray}
b \leftarrow \neg\!\textbf{}\bK a \label{f:loop2}
\end{eqnarray}
plus the above rule~\eqref{f:loop1} has now two rules defining atoms $a$ and $b$ in terms of the presence of those same atoms in all the stable models.
To solve this kind of cyclic interdependence, the original semantics by~\citeN{Gelfond91} (G91) considered different alternative \emph{world views} or sets of stable models.
In the case of program \eqref{f:loop1}-\eqref{f:loop2}, G91 yields two alternative world views\review{\footnote{\reviewcol For the sake of readability, sets of propositional interpretations are embraced with $\sset{ \ }$ rather than $\set{\ }$.}}{rev1.4}, $[\{a\}]$ and $[\{b\}]$, each one containing a single stable model, and this is also the behaviour obtained in the remaining approaches developed later on.
As noted by~\cite{Truszczynski11}, the feature that made G91 unconvincing, though, was the generation of self-supported world views.
A prototypical example for this effect is the epistemic program consisting of the single rule:
\begin{eqnarray}
a \leftarrow \bK a \label{f:self}
\end{eqnarray}
whose world views under G91 are $[\emptyset]$ and $[\{a\}]$.
The latter is considered as counter-intuitive by all authors\footnote{This includes Gelfond himself, who proposed a new variant in~\cite{Gelfond11} motivated by this same example and further modified this variant later on in~\cite{Kahl15}.} because it relies on a self-supported derivation: $a$ is derived from $\bK a$ by rule~\eqref{f:self}, but the only way to obtain $\bK a$ is rule \eqref{f:self} itself.
Recently,~\citeN{CabalarFF2019faeel} proposed to characterise these unintended world views by  extending the notion of \emph{unfounded sets}~\cite{GelderRS91} from standard disjunctive logic programs~\cite{LeoneRS97} to the case of epistemic logic programs.
In that work, the authors also provided a new semantics, called \emph{Founded Autoepistemic Equilibrium Logic} (FAEEL), that fulfills that requirement.
In fact, FAEEL-world views are precisely those G91-world views that are founded, that is, those that do not admit any unfounded~set.

On the other hand,
it is obvious that programs without epistemic cycles (i.e. cycles involving epistemic literals) cannot have self-supported derivations.
In this sense, one could expect that proposals that tried to get rid of G91 self-supported derivations coincided with the latter, at least, for epistemically acyclic programs.
However,~\cite{LK18b} have recently pointed out that this is not the case: for instance, while in G91, (purely) epistemic constraints always remove world views, this does not hold in other semantics.
\citeN{Watson2000} and~\citeANP{Cabalar2018nmr}~\citeNN{Cabalar2018nmr,CabalarFF2019splitting} went a step farther defining a property called \emph{epistemic splitting} which, not only defines an intuitive behaviour for stratified epistemic specifications, but also extends the \review{splitting theorem, well-known for autoepistemic logic~\cite{DBLP:journals/fuin/GelfondP92} and standard logic programs~\cite{LifschitzT94}, to the case of epistemic logic programs.}{rev2.2}
For instance, if we consider a program consisting of rules~\eqref{f:loop1}-\eqref{f:loop2} plus
\begin{eqnarray}
c \leftarrow \bK a \label{f:splitting}
\end{eqnarray}
we may expect to obtain the world views $[\{a,c\}]$ and $[\{b\}]$ resulting from adding the atom~$c$ only to the belief sets of the world view that satisfies~$\bK a$.
This property is known to be satisfied by the G91 semantics, but surprisingly not for those that tried to correct its self-supported problem~\cite{Cabalar2018nmr,CabalarFF2019splitting}.

\review{The major contribution of this paper is the proof that FAEEL satisfies the epistemic splitting property as defined in~\cite{Cabalar2018nmr,CabalarFF2019splitting}.
Joining this result with the already known fact that this semantics also satisfies the foundedness property shows that FAEEL is a solid candidate to serve as a semantics of epistemic logic programs.
A second contribution of this paper is the introduction of a logic that we call \emph{Founded Epistemic Equilibrium Logic}~(FEEL) and which can be intuitively seen as the combination of the Equilibrium Logic with the modal logic~S5.
For the sake of comparative, FAEEL corresponds to the combination of the Equilibrium Logic with the Moore's Autoepistemic Logic~(AEL;~\citeNP{Moore85}).
In this sense, FEEL is the combination of a non-monotonic logic with a monotonic one, while FAEEL is the combination of two non-monotonic logics, a fact that makes FEEL much easier to study.
This bring us to the third contribution of the paper: FAEEL world views can be precisely characterised as those G91 world views that are at the same time FEEL world views. This allows us to study FAEEL properties by studying them independently in FEEL and G91 and then combining their results.
This is precisely the methodology used in proving the epistemic splitting theorem for FAEEL.}{rev1.2}

The rest of the paper is organised as follows.
Section~\ref{sec:background}  revisits the background knowledge about equilibrium logic, epistemic specifications, the epistemic splitting property and FAEEL necessary for the rest of the paper.
Section~\ref{sec:feel} introduces FEEL and studies the relation between this logic and FAEEL.
In Section~\ref{sec:splitting.proof}, we study the epistemic splitting property in FEEL and FAEEL and, in Section~\ref{sec:related},
we discuss other existent approaches to epistemic logic programs.
Finally, Section~\ref{sec:conclusions} concludes the paper.

\section{Background}
\label{sec:background}

We start by recalling the basic definitions needed for the rest of the paper.
Given a set of \emph{atoms}~$\at$, an (epistemic) formula is defined according to the following grammar:
\[
\fF ::= \bot \mid a \mid \fF_1 \wedge \fF_2 \mid \fF_1 \vee \fF_2 \mid \fF_1 \to \fF_2 \mid \bL \varphi
\qquad\text{ for any atom } a \in \at. 
\]
In our context, the epistemic reading of $\bL \psi$ is that ``$\psi$ is one of the agent's beliefs.''
Thus, a formula $\varphi$ is said to be \emph{subjective} if all its atom occurrences (having at least one) are in the scope of~$\bL$.
Analogously, $\varphi$ is said to be \emph{objective} if $\bL$ does not occur in $\varphi$.
For instance, $\neg\!\bL a \vee \bL b$ is subjective, $\neg a \vee b$ is objective and $\neg a \vee \bL b$ none of the two.
\review{Given a formula $\varphi$, by $\Atoms(\varphi)$ we denote the set of all atoms occurring in $\varphi$.
For instance, $\Atoms(\neg a \vee \bL b) = \set{a,b}$.}{rev1.6}
As usual
we define the following derived operators:
$\fF \leftrightarrow \fG \eqdef (\fF \to \fG) \wedge (\fG \to \fF)$,
\ $(\fF \leftarrow \fG) \eqdef (\fG \to \fF)$,
\ $\neg \fF \eqdef (\fF \to \bot)$
and
\ $\top \eqdef \neg \bot$.
An \emph{(epistemic) theory} is a (possibly infinite) set of formulas as defined above
and an \emph{objective theory} is a theory whose formulas are objective.
We write $\Atoms(\varphi)$ to represent the set of atoms occurring in any formula~$\varphi$ and $\Atoms(\Gamma)$ to represent the set of atoms occurring in any theory~$\Gamma$.
Recall that~\cite{Gelfond91} included a second subjective operator~$\M$ such that $\M l$ is readed as ``the agent believes that $l$ is possible.''
In this paper, we assume here that $\M \varphi$ is just an abbreviation\footnote{Several interpretations of $\M$ are possible in the logics considered in this paper depending on the level of foundedness that it is expected to satisfy.
We limit ourselves here to the simplest interpretation of $\M$, leaving other interpretations for a more detailed discussion in the future.} for $\neg\K \neg \varphi$.

\subsection{Equilibrium Logic and the Stable Models Semantics}
\label{sec:ht}

A \emph{propositional interpretation} $T$ is a set of atoms $T \subseteq \at$. 
%For readability sake, we will write classical interpretations as underlined strings so that, for instance, given $\at=\{a,b,c\}$, the strings $\str{abc}$ and $\str{ac}$ stand for $\{abc\}$ and $\{a,c\}$; $\stre$ stands for the empty interpretation.
%
We write $T \models \varphi$ to represent that $T$ classically satisfies formula $\varphi$. 
An \emph{\htinterpretation} is a pair $\tuple{H,T}$ (respectively called ``here'' and ``there'') of propositional interpretations such that $H \subseteq T \subseteq \at$; it is said to be \emph{total} when $H=T$.
We write $\tuple{H,T} \models \varphi$ to represent that $\tuple{H,T}$ \emph{satisfies} a formula~$\varphi$ under the recursive conditions:
\begin{itemize}[ topsep=2pt]
\item $\tuple{H,T} \not\models \bot$ 
\item $\tuple{H,T} \models p$ iff $p \in H$ 
\item $\tuple{H,T} \models \varphi \wedge \psi$ iff $\tuple{H,T} \models \varphi$ and $\tuple{H,T} \models \psi$
\item $\tuple{H,T} \models \varphi \vee \psi$ iff $\tuple{H,T} \models \varphi$ or $\tuple{H,T} \models \psi$
\item $\tuple{H,T} \models \varphi \to \psi$ iff both (i) $T \models \varphi \to \psi$ and (ii) $\tuple{H,T} \not\models \varphi$ or $\tuple{H,T} \models \psi$
\end{itemize}
As usual, we say that $\tuple{H,T}$ is a \emph{model} of a theory~$\Gamma$,
in symbols $\tuple{H,T} \models \Gamma$, iff $\tuple{H,T} \models \varphi$
for all $\varphi \in \Gamma$.
It is easy to see that $\tuple{T,T} \models \Gamma$ iff $T \models \Gamma$ classically.
For this reason, we will identify $\tuple{T,T}$ simply as $T$ and will use `$\models$' indistinctly.
% By $\CL[\Gamma]$ we denote the set of all classical models of $\Gamma$.
% %
Interpretation $\tuple{T,T}=T$ is a \emph{stable (or equilibrium) model} of a theory $\Gamma$ iff $T \models \Gamma$ and there is no $H\subset T$ such that $\tuple{H,T} \models \Gamma$.
We write $\SM[\Gamma]$ to stand for the set of all stable models of $\Gamma$.
% %
% Note that $\SM[\Gamma] \subseteq \CL[\Gamma]$ by definition.

%%%%%%%%%%%%%%%%%%%%%%%%%%%%%%%%%%%%%%%%%%%%%%
\subsection{G91 semantics for epistemic theories}
\label{sec:g91}
%%%%%%%%%%%%%%%%%%%%%%%%%%%%%%%%%%%%%%%%%%%%%%

To represent the agent's beliefs, we will use a set $\wv$ of propositional interpretations.
We call \emph{belief set} to each element $I \in \wv$ and \review{\emph{belief view} the whole set~$\wv$.}{rev1.8}
The difference between belief and knowledge is that the  former may not hold in the real world.
Thus, satisfaction of formulas will be defined with respect to an interpretation $I \subseteq \at$, possibly $I \not\in \wv$, that accounts for the real world: the pair $\kdint{I}{\wv}$ is called \emph{belief interpretation} (or interpretation in modal logic KD45).
Modal satisfaction is also written \mbox{$\kdint{I}{\wv} \models \varphi$} (ambiguity is removed by the interpretation on the left) and follows the conditions:
\begin{itemize}[ topsep=2pt]
\item $\kdint{I}{\wv} \not\models \bot$,
\item $\kdint{I}{\wv} \models a$ iff $a \in I$, for any atom $a \in \at$,
\item $\kdint{I}{\wv} \models \psi_1 \wedge \psi_2$ iff $\kdint{I}{\wv} \models \psi_1$ and $\kdint{I}{\wv} \models \psi_2$,
\item $\kdint{I}{\wv} \models \psi_1 \vee \psi_2$ iff $\kdint{I}{\wv} \models \psi_1$ or $\kdint{I}{\wv} \models \psi_2$,
\item $\kdint{I}{\wv} \models \psi_1 \to \psi_2$ iff $\kdint{I}{\wv} \not\models \psi_1$ or $\kdint{I}{\wv} \models \psi_2$, and
\item $\kdint{I}{\wv} \models \bL \psi$ iff $\kdint{J}{\wv} \models \psi$ for all $J \in \wv$.
\end{itemize}
Notice that implication here is classical, that is, $\varphi \to \psi$ is equivalent to $\neg \varphi \vee \psi$ in this context.
A belief interpretation $\kdint{I}{\wv}$ is a \emph{belief model} of $\Gamma$ iff $\kdint{J}{\wv} \models \varphi$ for all $\varphi \in \Gamma$ and all $J \in \wv \cup \set{I}$.
We say that $\wv$ is an  \emph{epistemic model} of $\Gamma$, and abbreviate this as $\wv \models \Gamma$, iff $\kdint{J}{\wv} \models \varphi$ for all $\varphi \in \Gamma$ and all $J \in \wv$.
%
% If this is the case, then the belief interpretation $\kdint{I}{\wv}$ with a distinguished $I$ is further a \emph{belief model} of $\Gamma$ if, additionally, $\kdint{I}{\wv} \models \varphi$ for all $\varphi \in \Gamma$. 
%
Belief models defined in this way correspond to modal logic KD45 whereas epistemic models correspond to S5.
% Epistemic models defined in this way correspond to modal logic S5, whereas belief models correspond to KD45.
% $\bL \varphi \to \varphi$ (anything believed holds in the real world) so $\bL$ represents knowledge under these models.

\begin{example}\label{ex:s5models}
Take the theory $\Gamma_{\ref{f:loop1}}=\{\neg\!\bL b \to a\}$ corresponding to rule \eqref{f:loop1}.
An epistemic model $\wv \models \Gamma_{\ref{f:loop1}}$ must satisfy: 
 $\tuple{\wv,J} \models \bL b$ or $\tuple{\wv,J} \models a$, for all $J \in \wv$.
We get three epistemic models from $\bL b$, $\sset{\{b\}}$, $\sset{\{a,b\}}$, and  $\sset{\{b\},\{a,b\}}$ and the rest of cases must force $a$ true, so we also get $\sset{\{a\}}$ and $\sset{\{a\},\{a,b\}}$.
In other words, $\Gamma_{\ref{f:loop1}}$ has the same epistemic models as $\bL b \vee \bL a$.
\qed
\end{example}
Note that rule \eqref{f:loop1} alone did not seem to provide any reason for believing $b$, but we got three epistemic models above satisfying $\bL b$.
Thus, we will be interested only in some epistemic models (that as usual we will call \emph{world views}) that minimize the agent's beliefs in some sense.
To define such a minimisation we rely on the following syntactic transformation that extend the one given by~\cite{Truszczynski11} by stating a explicit signature in which it is applied.
The explicit signature will be useful later on to define the epistemic splitting property.

\begin{definition}[Subjective reduct]
The \emph{subjective reduct} of a theory $\Gamma$ with respect to a set of belief views~$\wv$ and a signature $U \subseteq \At$,
also written $\Gamma^\wv_U$, is obtained by replacing each maximal subjective formula of the form $\bK \varphi$ with $\Atoms(\varphi) \subseteq U$ by $\top$ if $\wv \models \bK \varphi$; or by $\bot$ otherwise.
When $U=\text{\rm \At}$ we just write $\Gamma^\wv$.\qed
\end{definition}

Finally, we impose a fixpoint condition where each belief set $I \in \wv$ is required to be a stable model of the reduct, \review{obtaining thus the G91 semantics.}{rev1.5}

\begin{definition}[G91 world view]
A belief view~$\wv$ is called a
\emph{G91-world view} of $\Gamma$ if and only if it satisfies
\mbox{$\wv=\SM[\Gamma^\wv]$}.\qed
\end{definition}

\begin{examplecont}{ex:s5models}\label{ex:s5models2}
Take any $\wv$ such that $\wv \models \bL b$.
Then, $\Gamma_{\ref{f:loop1}}^\wv = \{\bot \to a\}$ with $\SM[\Gamma_{\ref{f:loop1}}^\wv]=[\emptyset]$.
\review{The empty set does not satisfy $\bL b$ so $\wv$ cannot be a G91-world view of~$\Gamma_{\ref{f:loop1}}$.}{rev1.9}
If $\wv \not\models \bL b$ instead, we get $\Gamma_{\ref{f:loop1}}^\wv = \{\top \to a\}$,
whose unique stable model is $\{a\}$.
As a result, $\wv=\sset{\{a\}}$ is  the unique G91-world view.
\qed
\end{examplecont}

\begin{examplecont}{ex:s5models2}\label{ex:s5models3}
Let now $\Gamma_{\ref{f:loop2}}=\{\neg\!\bL b \to a\,,\, \neg\!\bL a \to b\}$ corresponding to rules \mbox{\eqref{f:loop1}-\eqref{f:loop2}}.
Take any $\wv$ such that $\wv \models \neg\!\bL a \wedge \bL b$.
Then, $\Gamma_{\ref{f:loop1}}^\wv = \{\bot \to a \,,\, \top \to b\}$ and we have that $\SM[\Gamma_{\ref{f:loop1}}^\wv]=[\set{b}]$.
Since $\wv=[\set{b}]$ satisfies $\neg\!\bL a \wedge \bL b$, this is a G91-world view of~$\Gamma_{\ref{f:loop2}}$.
If $\wv \models \bL a \wedge \neg\bL b$ instead, we get $\Gamma_{\ref{f:loop1}}^\wv = \{\top \to a\,,\, \bot \to b\}$,
whose unique stable model is $\{a\}$.
As a result, $\wv=\sset{\{a\}}$ is also the other G91-world view of~$\Gamma_{\ref{f:loop2}}$.
To see that there is not any other \review{world views}{rev1.9}, note that $\wv \models \neg\!\bL a \wedge \neg\!\bL b$
implies that
$\Gamma_{\ref{f:loop2}}^\wv = \{\top \to a \,,\, \top \to b\}$ and $\SM[\Gamma_{\ref{f:loop2}}^\wv]=[\set{a,b}]$.
So this cannot a G91-world view.
Similar, it can be checked that no world view can satisfy $\bL a \wedge \bL b$.
\qed
\end{examplecont}

\begin{example}\label{ex:self-supporting.rule}
Take now the theory $\Gamma_{\ref{f:self}}=\{\bL a \to a\}$ corresponding to rule \eqref{f:self}.
If $\wv \models \bL a$ we get $\Gamma_{\ref{f:self}}^\wv = \set{ \top \to a }$ and $\SM[\Gamma_{\ref{f:self}}^\wv]=\{a\}$ so $\wv=\sset{\{a\}}$ is a G91-world view.
If $\wv \not\models \bL a$, the reduct becomes $\Gamma_{\ref{f:self}}^\wv = \set{ \bot \to a }$, a classical tautology with unique stable model $\emptyset$.
As a result, $\wv=\sset{\emptyset}$ is the second G91-world view of this theory.
\qed
\end{example}

\subsection{Epistemic Specifications and Epistemic Splitting}
%%%%%%%%%%%%%%%%%%%%%%%%%%%%%%%%%%%%%%%%%%%%%%%%%%%%%%%%%%%%%%%%%%%%%%%%%%%%%%%%%%%%%%

In this section, we recall the formal definition of the Epistemic Splitting property.
For the motivation of the interest of this property we refer to~\cite{Cabalar2018nmr,CabalarFF2019splitting}.
Let start by introducing a particular class of theories that correspond to the syntax of epistemic specifications or \emph{(epistemic logic) programs}.
Given a set of atoms $S \subseteq \At$, by
\mbox{$\neg S \eqdef \setm{ \neg a }{ a \in S }$}
and
$\neg\neg S \eqdef \setm{ \neg\neg a }{ a \in S }$
we respectively denote the set resulting of preapend one or two occurrences of the default negation operator $\neg$ to every atom in $S$.
An \emph{objective literal} is either an atom or a truth constant\footnote{For a simpler description of program transformations, we allow truth constants with their usual meaning.},
that is \mbox{$a \in \At \cup \{\top,\bot\}$}, or the result of preapend one or two default negation,
\mbox{$\neg a$}.
By $\Litobj \eqdef \At \cup \neg \At \cup \neg\neg \at \cup \set{\top,\bot,\neg\top,\neg\bot,\neg\neg\top,\neg\neg\bot}$
we denote the set of all objective literals.
A \emph{subjective literal}  is any expression of the form $\K l$, $\neg \K l$ or $\neg\neg \K l$ with $l \in \Litobj$ any objective literal.
A \emph{literal} is either an objective or subjective literal.

A \emph{rule} $r$ is an implication of the form:
\begin{gather}
a_1 \vee \dots \vee a_n \leftarrow L_1 \wedge \dots \wedge L_m
	\label{eq:rule}
\end{gather}
with $n\geq 0$ and $m\geq 0$, where each $a_i \in \At$ is an atom and each $L_j$ a literal.
The left hand disjunction of \eqref{eq:rule} is called the rule \emph{head} and abbreviated as $\Head(r)$.
When $n=0$, it corresponds to $\bot$ and $r$ is called a \emph{constraint}.
The right hand side of \eqref{eq:rule} is called the rule \emph{body} and abbreviated as $\Body(r)$.
As usual,
we define $\Bodyp(r)$ and $\Bodyn(r)$ as the conjunction of all positive and negative literals in $\Body(r)$, respectively, so that $\Body(r) \equiv \Bodyp(r) \wedge \Bodyn(r)$.
We further define $\Bodyr(r)$ and $\Bodym(r)$ as the conjunction of all objective and subjective literals in $\Body(r)$, respectively, so that $\Body(r) \equiv \Bodyr(r) \wedge \Bodym(r)$.
We also define $\Body^x_y(r) \eqdef \Body^x(r) \cap \Body_y(r)$
with $x \in \set{+,-}$ and $y \in \set{\mathit{obj},\mathit{sub}}$.
By abuse of notation, we will also use sometimes $\Body^x$, $\Body_y$ and $\Body^x_y$ as the set of literals occurring in those formulas.
When $m=0$, the body corresponds to $\top$ and $r$ is called a \emph{fact} (in this case, the body and the arrow symbol are usually omitted).
A rule is called \emph{objective} if all literals occurring in it are objective.
A \emph{program} $\Pi$ is a (possibly infinite) set of rules
and an \emph{objective program} is a program where all its rules are objective.
We are now ready to recall the epistemic splitting property:

\begin{definition}[Epistemic splitting set]\label{def:splitting.set}
A set of atoms \mbox{$U \subseteq \At$} is said to be an \emph{epistemic splitting set} of a program $\Pi$ if for any rule $r$ in $\Pi$ one of the following conditions hold
\begin{enumerate}[ label=(\roman*), leftmargin=18pt]
\item $\Atoms(r) \subseteq U$,
    \label{item:1:def:splitting}
% \item $\Atoms(r) \cap U = \emptyset$, or
%     \label{item:2:def:splitting}
\item 
% $\Atoms(\Bodym(r)) \subseteq U$\\
% and
$(\Atoms(\Bodyr(r) \cup \Head(r)) ) \cap U = \emptyset$%
\label{item:3:def:splitting}%
\end{enumerate}
We define \review{a \emph{splitting}}{rev1.11} of 
$\Pi$ as a pair $\tuple{B_U(\Pi),T_U(\Pi)}$ satisfying $B_U(\Pi) \cap T_U(\Pi) = \emptyset$, $B_U(\Pi) \cup T_U(\Pi) = \Pi$, all rules in $B_U(\Pi)$ satisfy (i) and all rules in $T_U(\Pi)$ satisfy (ii).
\qed
\end{definition}

\noindent 
With respect to the original definition of splitting set, we can see that the condition for the top program $\Atoms(\Head(r)) \cap U = \emptyset$
was replaced by the new condition~(ii),
which
% This conditions allow rules in the top part that may only have atoms in $U$ inside a subjective literal.
intuitively means that the top program may only refer to atoms~$U$ in the bottom through epistemic operators.
Another observation is that the definition of $B_U(\Pi)$ and $T_U(\Pi)$ is kept \mbox{non-deterministic} in the sense that some rules can be arbitrarily included in one set or the other.
These rules correspond to subjective constraints on atoms in~$U$, since these are the only cases that may satisfy conditions~(i) and~(ii) simultaneously.
Then, the idea is similar as in splitting a regular program: first we compute the world views of the bottom program $B_U(\Pi)$ and for each one we compute the world views of the top program after simplifying in it the subjective literals fixed by the bottom part.
Formally, given an epistemic splitting set~$U$ for a program~$\Pi$ and belief view~$\wv$, we define $E_U(\Pi,\wv) \eqdef T_U(\Pi)^\wv_U$, that is, we make the subjective reduct of the top with respect to $\wv$ and signature $U$.

\begin{definition}\label{def:splitting}
Given a semantics~$\sS$,
a pair $\tuple{\wv_b,\wv_t}$ is said to be a \emph{$\sS$-solution} of $\Pi$ with respect to an epistemic splitting set $U$ if $\wv_b$ is a $\sS$-world view of $B_U(\Pi)$ and $\wv_t$ is a $\sS$-world view of $E_U(\Pi,\wv_b)$.\qed
\end{definition}

Notice that this definition depends on a particular semantics~$\sS$ in the sense that each alternative semantics for epistemic specifications will define its own solutions for a given $U$ and~$\Pi$.
In particular, in this paper, we will consider five instantiations of this Definition~\ref{def:splitting} with semantics $\sS \in \set{G91, FAEEL, FEEL, EEL, AEEL }$.
Besides the already mentioned G91, Founded Autoepistemic EEL (FAEEL) and Founded EEL (FEEL) semantics,\review{\footnote{\reviewcol To avoid possible confusions between FAEEL, FEEL, AEEL and EEL, we will sometimes write them as Founded Autoepistemic EEL, Founded EEL, Autoepistemic EEL and EEL respectively, when they occur in the same sentence.}}{rev1.3} we will also consider the EEL and AEEL semantics from~\cite{CerroHS15}.

\begin{examplecont}{ex:s5models3}\label{ex:s5models4}
Back to our example,
let now $\Pi_{\ref{f:splitting}}$ be the program consisting of rules \mbox{\eqref{f:loop1}-\eqref{f:loop2} and~\eqref{f:splitting}}.
Then, we can see that $U=\set{a,b}$ is an epistemic splitting set of~$\Pi_{\ref{f:splitting}}$ and that it satisfies
\mbox{$B_U(\Pi_{\ref{f:splitting}}) = \set{\eqref{f:loop1}-\eqref{f:loop2}}$}
and
\mbox{$T_U(\Pi_{\ref{f:splitting}}) = \set{\eqref{f:splitting}}$}.
Furthermore, we have already seen that $B_U(\Pi_{\ref{f:splitting}})$ corresponds to the theory~$\Gamma_{\ref{f:loop2}}$
which has the following two G91-world views: $\sset{\set{a}}$ and $\sset{\set{b}}$.
Then, we can see that 
\mbox{$E_U(\Pi_{\ref{f:splitting}},{\sset{\set{a}}}) = \set{ c \leftarrow \top }$}
has a unique G91-world view $\sset{\set{c}}$
and that
\mbox{$E_U(\Pi_{\ref{f:splitting}},{\sset{\set{b}}}) = \set{ c \leftarrow \bot }$}
has the unique G91-world view $\sset{\emptyset}$.
As a result, we have two G91-solutions of $\Pi_{\ref{f:splitting}}$ with respect to $\set{a,b}$:
$\tuple{\sset{\set{a}},\sset{\set{c}}}$
and
$\tuple{\sset{\set{b}},\sset{\emptyset}}$.
It is also easy to check that 
\mbox{$B_U(\Pi_{\ref{f:splitting}})$}
has two G91-world views, $\sset{\set{a,c}}$ and $\sset{\set{b}}$ that can be obtained by composing the two above solutions.
\qed
\end{examplecont}

In the general case, the world views for the global program are reconstructed by the following operation:
$$\wv_b \sqcup \wv_t \ \  = \ \ \setm{I_b \cup I_t}{ I_b \in \wv_b  \text{ and } I_t \in \wv_t  }$$
(remember that both the bottom and the top may produce multiple world views, depending on the program and the semantics we choose).
For any set of atoms~$U \subseteq \at$ and belief view~$\wv$, we also define the restriction of $\wv$ to $U$ as $\restr{\wv}{U} \eqdef \setm{ I \cap U }{ I \in \wv }$.
Furthermore, we also define the complement of a set of atoms as $\overline{U} \eqdef \at \setminus U$.

\begin{property}[Epistemic splitting]\label{propy:epistemic.splitting}
A semantics~$\sS$ satisfies \emph{epistemic splitting} if for any epistemic splitting set $U$ of any program $\Pi$: $\wv$ is a $\sS$-world view of $\Pi$ iff there is a $\sS$-solution $\tuple{\wv_b,\wv_t}$ of $\Pi$ with respect to $U$ such that $\wv=\wv_b \sqcup \wv_t$. \qed
\end{property}

\begin{theorem}\label{thm:epistemic.splitting.g91}
Semantics G91 satisfies epistemic splitting.
Furthermore, if $\wv$ is a G91-world view of some program~$\Pi$ with respect to some splitting set~$U$,
then $\tuple{\restr{\wv}{U},\restr{\wv}{\overline{U}}}$ is a G91-solution of~$\Pi$ and it satisfies that $\wv = \restr{\wv}{U} \sqcup \restr{\wv}{\overline{U}}$.\qed
\end{theorem}

Theorem~\ref{thm:epistemic.splitting.g91} was proved in~\cite[Main Theorem]{CabalarFF2019splitting}.
Note that there, it is only stated that G91 satisfies epistemic splitting, however it is easy to see that the second part of the statement was proved as an auxiliary result inside the proof of that theorem.
We decided to explicitly state this result as it will be useful for proving that FAEEL satisfies epistemic splitting.

\subsection{Founded Autoepistemic Equilibrium Logic}

We recall now the semantics of \emph{Founded Autoepistemic Equilibrium Logic} (FAEEL) from~\cite{CabalarFF2019faeel}.
The basic idea is an elaboration of the belief (or KD45) interpretation $\kdint{I}{\wv}$ already seen but replacing belief sets by HT pairs.
Thus, the idea of \emph{belief view} $\wv$ is extended to a non-empty set of HT-interpretations $\wv=\{\tuple{H_1,T_1}, \dots, \tuple{H_n,T_n}\}$ and  say that $\wv$ is \emph{total} when $H_i=T_i$ for all of them, coinciding with the form of belief views $\wv=\{T_1, \dots, T_n\}$ we had so far.
Similarly, a \emph{belief interpretation} is now redefined as $\kdint{\tuple{H,T}}{\wv}$, or simply $\kdint{H,T}{\wv}$, where $\wv$ is a belief view and $\tuple{H,T}$ stands for the real world, possibly not in $\wv$.
A belief interpretation~$\kdint{H,T}{\wv}$ is called \emph{total} iff both $\tuple{H,T}$ and $\wv$ are total.
Next, the satisfaction relation is defined as a combination of modal logic KD45 and HT.
A belief interpretation $\cI = \kdint{H,T}{\wv}$ satisfies a formula $\varphi$, written $\cI \models \varphi$, iff:
\begin{itemize}[topsep=2pt]
\item $\cI \not\models \bot$,
\item $\cI \models a$ iff $a \in H$, for any atom $a \in \at$,
\item $\cI \models \psi_1 \wedge \psi_2$ iff $\cI \models \psi_1$ and $\cI \models \psi_2$,
\item $\cI \models \psi_1 \vee \psi_2$ iff $\cI \models \psi_1$ or $\cI \models \psi_2$,
\item $\cI \models \psi_1 \to \psi_2$ iff both: (i) $\cI \not\models \psi_1$ or $\cI \models \psi_2$;
and (ii) $\kdint{T}{\wv^t} \not\models \psi_1$ or
$\kdint{T}{\wv^t} \models \psi_2$,\\\hspace*{2.25cm} where\review{\footnote{\reviewcol Note that $\wv^t$ is a belief view as defined in Section~\ref{sec:g91}.}}{rev1.12} $\wv^t=\{T_i \mid \tuple{H_i,T_i} \in \wv\}$.
\item $\cI \models \bL \psi$ iff $\kdint{H_i,T_i}{\wv} \models \psi$ for all $\tuple{H_i,T_i} \in \wv$.
\end{itemize}
A belief interpretation $\kdint{H,T}{\wv}$ is called a \emph{belief model} of a theory~$\Gamma$ iff $\kdint{H_i,T_i}{\wv} \models \varphi$ for all \htinterpretation $\tuple{H_i,T_i} \in \wv \cup\{\tuple{H,T}\}$ and all $\varphi \in \Gamma$.
Given theories $\Gamma$ and $\Gamma'$, we write $\Gamma \models \Gamma'$ when $\kdint{H,T}{\wv} \models \Gamma$ implies $\kdint{H,T}{\wv} \models \Gamma'$ for all belief interpretations.
We write $\Gamma \equiv \Gamma'$ iff $\Gamma \models \Gamma'$ and $\Gamma' \models \Gamma$.
Furthermore, when $\Gamma$ or $\Gamma'$ are singletons we may omit the brackets around their unique formula.

\review{

Recall that the negation of a formula $\neg\varphi$ is defined as an abbreviation for the implication $\varphi\to\bot$.
The following result is immediate from the above definition plus the persistence property proved in~\cite{CabalarFF2019faeel} (Proposition~1) and explicitly states the evaluation of negation:

\begin{proposition}
Given a belief interpretation~$\cI= \kdint{H,T}{\wv}$ and a formula~$\varphi$, it follows that $\cI \models \neg\varphi$ iff $\kdint{T}{\wv^t} \not\models \varphi$.\qed
\end{proposition}
}{rev3.2}

\review{As recalled in Section~\ref{sec:ht}, stable models correspond to a class of HT-models called equilibrium models, that is, total minimal models.
Similarly, we define now equilibrium belief models as total minimal belief models with respect to the following order relation:}{rev1.12}

\begin{definition}\label{def:int.prec}
We define the partial order $\cI' \preceq \cI$ for belief interpretations $\cI'=\kdint{H',T'}{\wv'}$ and $\cI=\kdint{H,T}{\wv}$ when the following three conditions hold:
\begin{enumerate}[ leftmargin=18pt , label=(\roman*) ]
\item $T' = T$ and $H' \subseteq H$, and

\item for every $\tuple{H_i,T_i} \in \wv$,
there is some $\tuple{H'_i,T_i} \in \wv'$,
with $H'_i \subseteq H_i$.

\item for every $\tuple{H'_i,T_i} \in \wv'$,
there is some $\tuple{H_i,T_i} \in \wv$,
with $H'_i \subseteq H_i$.\qed
\end{enumerate}
\end{definition}
As usual, $\cI' \prec \cI$ means $\cI' \preceq \cI$ and $\cI' \neq \cI$.

\begin{definition}
A total belief interpretation $\cI=\kdint{T}{\wv}$ is said to be an \emph{equilibrium belief model} of some theory~$\Gamma$ iff $\cI$ is a belief model of $\Gamma$ and there is no other belief model $\cI'$ of $\Gamma$ such that $\cI' \prec \cI$.\qed
\end{definition}
By $\MEQ[\Gamma]$ we denote the set of equilibrium belief models of $\Gamma$.
As a final step, we impose a fixpoint condition to minimise the agent's knowledge as follows.
\begin{definition}\label{def:au-eqmodel}
A total belief view~$\wv$ is called an \emph{autoepistemic equilibrium model} or \emph{\mbox{FAEEL-world} view} of~$\Gamma$
iff:
\begin{IEEEeqnarray*}{c+x*}
\wv \ \ = \ \ \setm{ T }{ \kdint{T}{\wv } \in \MEQ[ \Gamma ] }
&\qed
\end{IEEEeqnarray*}
\end{definition}

\begin{Theorem}[Main Theorem in~\citeNP{CabalarFF2019faeel}]{\label{thm:g91}}
For any theory $\Gamma$, its FAEEL-world views are exactly its founded\footnote{For space reasons we omit here the definition of \emph{founded world view} and refer the reader to~\cite{CabalarFF2019faeel}. Intuitively, a world view is founded if all atoms that are true in all its belief set can be derived without cyclic references.
If we omit there founded, we obtain that every FAEEL-world view is also a G91-world view, but not necessarily vice-versa.} G91-world views of $\Gamma$.\qed
\end{Theorem}

\section{Founded Epistemic Equilibrium Logic}
\label{sec:feel}

Founded EEL is similar to Founded Autoepistemic EEL, but without the minimisation of knowledge.
Technically, this makes \review{Founded EEL}{rev1.13} simpler in two distinct ways: (i) it directly uses belief views instead of belief interpretations and, as a result, (ii) it lacks the autoepistemic fixpoint condition (Definition~\ref{def:au-eqmodel}).
Note that, as mentioned in the introduction Founded EEL can be seen as the combination of the stable model semantics with the modal logic S5, while Founded Autoepistemic EEL would be the combination of the stable model semantics with Moore's Autoepistemic Logic.
\review{In this sense, (i) is a direct consequence of the fact that Moore's Autoepistemic Logic is defined in terms of modal logic KD45 instead of S5.
In its turn, (ii) is a consequence of the fact that S5 is a monotonic logic, and thus Founded EEL do not need the autoepistemic fixpoint condition (Definition~\ref{def:au-eqmodel}) that Founded Autoepistemic EEL inherits from Moore's Autoepistemic Logic.}{rev1.13}

Formally, a belief view~$\wv$ is called an \emph{epistemic model} of a theory~$\Gamma$, in symbols $\wv \models \Gamma$ iff $\kdint{H_i,T_i}{\wv} \models \varphi$ for all \review{\htinterpretation}{rev1.14} $\tuple{H_i,T_i} \in \wv$ and all $\varphi \in \Gamma$.
Given theories $\Gamma$ and $\Gamma'$, we write $\Gamma \modelsfeel \Gamma'$ when $\kdint{H,T}{\wv} \models \Gamma$ implies $\kdint{H,T}{\wv} \models \Gamma'$ for all belief interpretations.
We write $\Gamma \equivfeel \Gamma'$ iff $\Gamma \modelsfeel \Gamma'$ and $\Gamma' \modelsfeel \Gamma$.
As above, when $\Gamma$ or $\Gamma'$ are singletons we may omit the brackets around their unique formula.

\begin{Proposition}[Persistence]{\label{prop:persistance}}
$\wv\models \varphi$ implies $\wv^t \models \varphi$.\qed
\end{Proposition}

\begin{proof}
Follows directly from Proposition~1 in~\cite{CabalarFF2019faeel}.
\end{proof}

The following order relation adapts Definition~\ref{def:int.prec} to the case of belief views.

\begin{definition}\label{def:int.prec.views}
Given belief views
$\wv_1$
and
$\wv_2$,
we write $\wv_1 \preceq \wv_2$
iff the following two condition hold:
\begin{enumerate}[ leftmargin=18pt , label=(\roman*) ]
\item for every $\tuple{H_2,T} \in \wv_2$,
there is some $\tuple{H_1,T} \in \wv_1$,
with $H_1 \subseteq H_2$.

\item for every $\tuple{H_1,T} \in \wv_1$,
there is some $\tuple{H_2,T} \in \wv_2$,
with $H_1 \subseteq H_2$.
\end{enumerate}
As usual, we write $\wv_1 \prec \wv_2$ iff $\wv_1 \preceq \wv_2$ and $\wv_1 \neq \wv_2$.\qed
\end{definition}

Then, equilibrium epistemic models are defined as usual:

\begin{definition}\label{def:eq.epistemic.models}
A total epistemic model $\wv$ of a theory~$\Gamma$ is said to be an \emph{epistemic equilibrium model} or \emph{\mbox{FEEL-world} view} iff there is no other epistemic model~$\wv'$ of $\Gamma$ such that $\wv' \prec \wv$.\qed
\end{definition}

The following observation establishes a relation between (equilibrium) belief models and (equilibrium) epistemic models similar to the existent between the standard modal logics~KD45 and~S5.
Recall that belief models correspond to the logical product of HT and KD45 while epistemic models come from the logical product of HT and S5.

\begin{observation}\label{obs:ht.d45models.are.s5models}
For any theory~$\Gamma$ and belief interpretation~$\cI = \tuple{\wv,H,T}$, the following statements hold:
\begin{enumerate}[ leftmargin=18pt , label=(\roman*) ]
\item If $\cI$ is a belief model of~$\Gamma$, then $\wv$ is a epistemic model of~$\Gamma$, and
\item If $\cI$ is an equilibrium belief model of~$\Gamma$, then $\wv$ is a equilibrium epistemic model of~$\Gamma$. \qed
\end{enumerate}
\end{observation}

From this observation the following relation between autoepistemic equilibrium models and epistemic equilibrium models can be established:

\begin{theorem}\label{prop:faeel->feel}
Every FAEEL-world view of any theory~$\Gamma$ is also an \mbox{FEEL-world} view of~$\Gamma$.\qed
\end{theorem}

In general the converse does not necessary holds as illustrated by the following example:

\begin{example}\label{ex:faeel.vs.feel}
Consider the program $\newprogram\label{prg:disjuntion}$ consisting of the single rule $a \vee b$.
This program has two stable models, $\set{a}$ and $\set{b}$ and accordingly a unique FAEEL-world view~$\sset{\set{a},\set{b}}$ which agrees with its unique G91-world view.
On the other hand,
this program has two extra \mbox{FEEL-world} views that are not
FAEEL-world views: $\sset{\set{a}}$ and $\sset{\set{b}}$.
These two extra models are clearly not well-justified in an epistemic sense as they respectively satisfy~$\bK a$ and~$\bK b$ with no evidence for that conclusion.
Finally, to see the difference between FEEL and modal logic S5, note that $\sset{\set{a,b}}$ is also an S5 model which is neither a FEEL nor a FAEEL-world view.\qed
\end{example}

This example also illustrates that the difference between FAEEL and FEEL can be formalised in terms of the \emph{supra-ASP} property introduced in~\cite{CabalarFF2019splitting} and recalled below: FAEEL satisfies supra-ASP~\cite[Proposition~3]{CabalarFF2019faeel} while FEEL does not.

\begin{property}[Supra-ASP]\label{property:supraASP}
A semantics~$\sS$ satisfies \emph{supra-ASP} if for any objective program $\Pi$ either~$\Pi$ has a unique $\sS$-world view $\wv=\SM[\Pi] \neq \emptyset$ or \review{both $\SM[\Pi]=\emptyset$ and $\Pi$ has no $\sS$-world view at all.}{rev1.17}\qed
\end{property}

On the other hand, the following result shows that we can obtain FAEEL-world views as the intersection of FEEL and G91-world views:

\begin{Theorem}{\label{thm:faeel<->feel+g91}}
For any theory~$\Gamma$, a belief view~$\wv$ is a FAEEL-world view iff (i) $\wv$ is a \mbox{FEEL-world} view and (ii) $\wv$ is a G91-world view.\qed
\end{Theorem}

\begin{examplecont}{ex:self-supporting.rule}\label{ex:self-supporting.rule3}
Back to~\mbox{$\Gamma_{\ref{f:self}}=\{\bL a \to a\}$}, recall that
this theory has two \mbox{G91-world} views:~$\sset{\emptyset}$ and $\sset{\set{a}}$.
It is easy to see that $\sset{\emptyset}$ is a \mbox{FEEL-world} view
and, from Theorem~\ref{thm:faeel<->feel+g91}, this implies that
this is also a FAEEL-world view.
Note that
there is no smaller belief view than~$\sset{\emptyset}$, so being a model is enough to show that is a \mbox{FEEL-world} view.
Note that, \review{by using Theorem~\ref{thm:faeel<->feel+g91},}{rev1.18} is much easier to see that 
$\sset{\emptyset}$ is a FAEEL-world view than directly using its definition since we would also need to check no other
\mbox{$T \not\in \sset{\emptyset}$}
satisfies
\mbox{$\kdint{T}{\sset{\emptyset}} \in \MEQ[\Gamma_{\ref{f:self}}]$}.
In this case, the only possibility is $\kdint{\set{a}}{\sset{\emptyset}}$ which fails because there is a smaller belief model
$\kdint{\emptyset,\set{a}}{\sset{\emptyset}}$
satisfying \mbox{$\bL a \to a$}.
On the other hand,
we can see that $\sset{\set{a}}$ is not a \mbox{FEEL-world} view and, thus neither FAEEL-world view, because $\sset{\tuple{\emptyset,\set{a}}}$ also satisfies
$\Gamma_{\ref{f:self}}$.
In this case, it is also easier to use
Theorem~\ref{thm:faeel<->feel+g91}
than using the FAEEL definition:
$\sset{\set{a}}$ is not a FAEEL-world view because 
$\cI' = \kdint{\set{a}}{\sset{\set{a}}} \not\in \MEQ[\Gamma_{\ref{f:self}}]$
and this is the case
because the smaller interpretation $\cI ''=\kdint{\set{a},\set{a}}{\sset{\tuple{\emptyset,\set{a}}}}$ also satisfies $\Gamma_{\ref{f:self}}$.
In particular, note that $\cI'' \not\models \bL a$
and, thus, clearly satisfies $\bL a \to a$.
\qed
\end{examplecont}

Example~\ref{ex:self-supporting.rule3} illustrates how Theorem~\ref{thm:faeel<->feel+g91} can be used to find FAEEL-world views using FEEL.
In general, this is much easier than directly applying their definition because belief views are simpler than belief interpretations
and because the autoepistemic fixpoint can be checked independently using the G91-semantics.

\section{Epistemic Splitting in Founded (Auto)Epistemic Equilibrium Logic}
\label{sec:splitting.proof}

Let us now study the epistemic splitting property in FEEL and FAEEL.
Let us start by stating a result analogous to Theorem~\ref{thm:epistemic.splitting.g91}, but this time for FEEL.

\begin{Theorem}[Epistemic splitting in FEEL]{\label{thm:epistemic.splitting.feel}}
FEEL satisfies the epistemic splitting property.
Furthermore, if $\wv$ is a \mbox{FEEL-world} view of some program~$\Pi$ with respect to some splitting set~$U$
then $\tuple{\restr{\wv}{U},\restr{\wv}{\overline{U}}}$ is a FEEL-solution of~$\Pi$ and it satisfies that $\wv = \restr{\wv}{U} \sqcup \restr{\wv}{\overline{U}}$.\qed
\end{Theorem}

The proof of Theorem~\ref{thm:epistemic.splitting.feel} is based in the following auxiliary results whose proof can be found in the supplementary material.

\begin{Proposition}{\label{prop:splitting.aux}}
Let $U \subseteq \at$ be some set of atoms and $\Pi = \Pi_1 \cup \Pi_2$ be a program such that
$\Atoms(\Pi_1) \subseteq U$ and $\Atoms(\Pi_2) \subseteq \overline{U}$.
Then, any belief view $\wv$ is an \mbox{FEEL-world} view of $\Pi$ iff
$\restr{\wv}{U}$ is an \mbox{FEEL-world} view of $\Pi_1$ and
$\restr{\wv}{\overline{U}}$ is an \mbox{FEEL-world} view of $\Pi_2$.\qed
\end{Proposition}

Intuitively, Proposition~\ref{prop:splitting.aux} says that, if we can split a program in a way that its two halves do not share atoms in common, then 
we can compute the world views of the whole program by combining the world views of each half.
\review{Furthermore, the following result
shows that we can check whether 
a belief view is an \mbox{FEEL-world} view by checking instead that 
that belief view is a \mbox{FEEL-world} view of a program obtained by simplifying the subjective literals in the rules of the top part accordingly to the belief view.}{rev1.19}

\begin{Proposition}{\label{prop:splitting.aux2}}
Let $\Pi$ be a program with epistemic splitting set~$U \subseteq \at$.
Then, any belief view~$\wv$ is an \mbox{FEEL-world} view of~$\Pi$
iff $\wv$ is an \mbox{FEEL-world} view of~$B_U(\Pi) \cup E_U(\Pi,\wv)$.\qed
\end{Proposition}

We can now join Propositions~\ref{prop:splitting.aux} and~\ref{prop:splitting.aux2} to show the following rewriting of epistemic splitting.

\begin{proposition}{\label{prop:epistemic.splitting.feel}}
Given any program $\Pi$ and epistemic splitting set~$U$ of~$\Pi$,
a belief view $\wv$ is an \mbox{FEEL-world} view of $\Pi$ iff 
$\restr{\wv}{U}$ is an \mbox{FEEL-world} view of $B_U(\Pi)$ and
$\restr{\wv}{\overline{U}}$ is an \mbox{FEEL-world} view of $\restr{\wv}{\overline{U}}$ of $E_U(\Pi,\restr{\wv}{U})$.\qed
\end{proposition}

\begin{proof}
First note that, from Proposition~\ref{prop:splitting.aux2},
it follows that
$\wv$ is an \mbox{FEEL-world} view of $\Pi$ iff 
$\wv$ is an \mbox{FEEL-world} view of~$B_U(\Pi) \cup E_U(\Pi,\wv)$.
Furthermore, we have
$E_U(\Pi,\wv) = E_U(\Pi,\restr{\wv}{U})$.
Therefore, we immediately can see that
$\wv$ is an \mbox{FEEL-world} view of $\Pi$ iff 
$\wv$ is an \mbox{FEEL-world} view of~$B_U(\Pi) \cup E_U(\Pi,\restr{\wv}{U})$.
Furthermore, by construction we have $\Atoms(B_U(\Pi)) \subseteq U$ and $\Atoms(E_U(\Pi,\restr{\wv}{U})) \subseteq \overline{U}$ and, from Proposition~\ref{prop:splitting.aux},
this implies that the latter holds iff
$\restr{\wv}{U}$ is an \mbox{FEEL-world} view of $B_U(\Pi)$ and
$\restr{\wv}{\overline{U}}$ is an \mbox{FEEL-world} view of $E_U(\Pi,\restr{\wv}{U})$.
\end{proof}

\begin{Proofof}{\ref{thm:epistemic.splitting.feel}}
Assume first that $\wv$ is an equilibrium epistemic model of $\Pi$.
Then, from Proposition~\ref{prop:epistemic.splitting.feel},
it follows that
$\restr{\wv}{U}$ is an \mbox{FEEL-world} view of $B_U(\Pi)$ and
$\restr{\wv}{\overline{U}}$ is an \mbox{FEEL-world} view of $\restr{\wv}{\overline{U}}$ of $E_U(\Pi,\restr{\wv}{U})$ and it is easy to check that \mbox{$\wv = \restr{\wv}{U} \sqcup \restr{\wv}{\overline{U}}$}.
The other way around, assume there are \mbox{FEEL-world} views $\wv_b$ of $B_U(\Pi)$ and $\wv_t$ of $E_U(\Pi,\wv_b)$ and let $\wv=\wv_b \sqcup \wv_t$.
Note that $\Atoms(B_U(\Pi)) \subseteq U$ implies that every interpretation~$T \in \wv_b$ satisfies $T \subseteq U$
and that $\Atoms(E_U(\Pi,\wv_b)) \subseteq \overline{U}$ implies that every $T \in \wv_t$ satisfies $T \subseteq \overline{U}$.
Hence, it follows that $\wv_b = \restr{\wv}{U}$ and $\wv_t = \restr{\wv}{\overline{U}}$ and the result follows directly from Proposition~\ref{prop:epistemic.splitting.feel}.
\end{Proofof}

We can now use Theorem~\ref{thm:epistemic.splitting.feel} in combination with Theorems~\ref{thm:epistemic.splitting.g91} and~\ref{thm:faeel<->feel+g91} to show that FAEEL also satisfies epistemic splitting.

\begin{mtheorem}[Epistemic splitting in FAEEL]{\label{thm:epistemic.splitting.faeel}}
FAEEL satisfies the epistemic splitting property.
Furthermore, if $\wv$ is a FAEEL-world view of some program~$\Pi$ with respect to some splitting set~$U$
then $\tuple{\restr{\wv}{U},\restr{\wv}{\overline{U}}}$ is a FAEEL-solution of~$\Pi$ and it satisfies that
$\wv = \restr{\wv}{U} \sqcup \restr{\wv}{\overline{U}}$.\qed
\end{mtheorem}

\begin{proof}
Assume first that there is FAEEL-solution~$\tuple{\wv_b,\wv_t}$ of some program~$\Pi$ with respect to some splitting set~$U$
and let $\wv = \wv_b \sqcup \wv_t$.
By definition, this implies that $\wv_b$ is a FAEEL-world view of $B_U(\Pi)$ which, from Theorem~\ref{thm:faeel<->feel+g91},
implies that $\wv_b$ is both a FEEL and a G91-world view of $B_U(\Pi)$.
Similarly, we can see that $\wv_t$ is both a FEEL and a G91-world view of $E_U(\Pi,\wv_b)$
and, thus, that $\tuple{\wv_b,\wv_t}$ is both a FEEL and a G91-solution of~$\Pi$ with respect~$U$.
From Theorems~\ref{thm:epistemic.splitting.g91} and~\ref{thm:epistemic.splitting.feel},
these two facts respectively imply that $\wv$ is both a FEEL and a G91-world view of~$\Pi$ which, from Theorem~\ref{thm:faeel<->feel+g91} again, implies that $\wv$ is a FAEEL-world view of~$\Pi$.

The other way around is analogous.
Assume now that $\wv$ is a FAEEL-world view of some program~$\Pi$ with splitting set~$U$.
Then, from Theorem~\ref{thm:faeel<->feel+g91}, it follows that $\wv$ is both a FEEL and a G91-world view.
From Theorems~\ref{thm:epistemic.splitting.g91} and~\ref{thm:epistemic.splitting.feel},
these two facts respectively imply that $\tuple{\restr{\wv}{U},\restr{\wv}{\overline{U}}}$ is a FEEL and a G91-solution of~$\Pi$ with respect to~$U$ and $\wv = \restr{\wv}{U} \sqcup \restr{\wv}{\overline{U}}$.
Finally, from Theorem~\ref{thm:faeel<->feel+g91} again, this implies that $\tuple{\restr{\wv}{U},\restr{\wv}{\overline{U}}}$ is also a FAEEL-solution.
\end{proof}

It is interesting to note that \review{for any semantics that satisfies}{rev1.20} epistemic splitting, thus FAEEL and~G91, constraints indented to remove world views are well-behaved:

\begin{property}[Subjective constraint monotonicity]
A semantics satisfies \emph{subjective constraint monotonicity} if, for any epistemic program $\Pi$ and any subjective constraint $r$, $W$ is a world view of $\Pi \cup \{r\}$ iff both $W$ is a world view of $\Pi$ and $W \models r$.\qed
\end{property}

\begin{theorem}[Theorem~2 in~\citeNP{CabalarFF2019splitting}]\label{thm:splitting->constraint.monotonicity}
Epistemic splitting implies subjective constraint monotonicity. \qed
\end{theorem}

Furthermore, this property also guarantees that, for \emph{epistemically stratified} programs, these semantics have at most a unique world view~(see Theorem~1 in~\citeNP{CabalarFF2019splitting}).
Another interesting property, not included in~\cite{CabalarFF2019splitting}, is that any semantics that satisfies epistemic splitting and the \emph{supra-ASP} property, necessary \review{coincides with the G91-semantics for the class}{rev1.21} of epistemic stratified programs.
Before tackling the notion of epistemic stratification, let us recall the epistemic dependence relation among atoms in a program $\Pi$ so that $dep(a,b)$ is true iff there is a rule $r \in \Pi$ such that $a \in \Atoms(\Head(r) \cup \Bodyr(r))$ and $b \in \Atoms(\Bodym(r))$.

\begin{definition}
We say that an epistemic program $\Pi$ is \emph{epistemically stratified} if we can assign an integer mapping $\lambda: \At \to \mathbb{N}$ to each atom such that
\begin{enumerate}[ leftmargin=18pt , label=(\roman*) ]
    \item \mbox{$\lambda(a) = \lambda(b)$} for any rule $r \in \Pi$ and atoms $a,b \in (\Atoms(r)\setminus \Bodym(r))$,
    \item \mbox{$\lambda(a)>\lambda(b)$} for any pair of atoms $a,b$ satisfying $dep(a,b)$.\qed
\end{enumerate}
\end{definition}

\begin{Theorem}{\label{thm:stratified.agrees.G91}}
Given any two semantics~$\sS$ and~$\sS'$ \review{that satisfy}{rev1.22} the epistemic splitting and supra-ASP properties and an epistemically stratified program~$\Pi$, one of following two condition hold:
\begin{enumerate}[ leftmargin=18pt , label=(\roman*) ]
\item $\Pi$ has neither $\sS$-world view nor $\sS'$-world view, or
\item $\Pi$ has exactly one $\sS$ and one $\sS'$-world view~$\wv$ which is the same in both semantics.
\end{enumerate}
\end{Theorem}

\begin{proof}
The proof follows by induction in the number of layers induced by the stratifications.
Note that, if a program has a unique layer, then it must be objective and, thus, the result follows directly from the supra-ASP property.
Otherwise, let $a \in \at$ such that there is no $b \in \at$ with
$\lambda(a) < \lambda(b)$
and let $U = \setm{ c \in \at }{ \lambda(c) < \lambda(a) }$.
Then, $U$ is an splitting set of~$\Pi$ 
and the epistemic splitting property tell us a belief view~$\wv$ is a $\sS$-world view of~$\Pi$ iff there are $\sS$-world views
$\wv_b$ of $B_U(\Pi)$ and $\wv_t$ of $E_U(\Pi,\wv_b)$.
Furthermore, $B_U(\Pi)$ and $E_U(\Pi,\wv_b)$ have less layers than~$\Pi$ so, by induction hypothesis, this holds iff
iff there are $\sS'$-world views
$\wv_b$ of $B_U(\Pi)$ and $\wv_t$ of $E_U(\Pi,\wv_b)$
iff
$\wv$ is a $\sS'$-world view of~$\Pi$.
\end{proof}

\begin{corollary}\label{cor:stratified.coincide}
For epistemically stratified programs, FAEEL and G91-world views coincide.\qed
\end{corollary}

Note that Corollary~\ref{cor:stratified.coincide} does not apply to FEEL because this semantics does not satisfy \mbox{supra-ASP} as illustrated by Example~\ref{ex:faeel.vs.feel}.
Recall also that, from Proposition~2 in~\cite{CabalarFF2019faeel}, we already knew that, for programs where all occurrences of $\bK$ are in the scope of negation, FAEEL and \mbox{G91-world} views coincide.
Corollary~\ref{cor:stratified.coincide} enlarges the class programs in which FAEEL and G91 coincide by including all those that are epistemically stratified.
This immediately arises the question whether FAEEL and G91 also coincide for programs without positive cycles involving epistemic literals.
The following result shows that this is indeed the case.
Formally, we define the positive epistemic dependence relation among atoms in a program $\Pi$ so that $dep^+(a,b)$ is true iff there is a rule $r \in \Pi$ such that $a \in \Atoms(\Head(r) \cup \Bodyr(r))$ and $b \in \Atoms(\Bodymp(r))$.

\begin{definition}[Epistemically tight program]
We say that an epistemic program $\Pi$ is \emph{epistemically tight} if we can assign an integer mapping
$\lambda: \at \to \mathbb{N}$ to each atom such that
\begin{enumerate}[ label=(\roman*), leftmargin=18pt]
    \item \mbox{$\lambda(a) = \lambda(b)$} for any rule $r \in \Pi$ and atoms $a,b \in (\Atoms(r)\setminus \Bodym(r))$,
    \item \mbox{$\lambda(a)>\lambda(b)$} for any pair of atoms $a,b$ satisfying $dep^+(a,b)$.\qed
\end{enumerate}
\end{definition}

\begin{definition}
Given an epistemic theory~$\Gamma$ and a belief view~$\wv$,
its \emph{negatively subjective reduct}, in symbols~$\Gamma^{\underline{\wv}}$,
is obtained by replacing each maximal subjective formula of the form $\neg\!\bK \varphi$ by~$\top$ if $\wv \not\models \bK \varphi$; or by $\bot$ otherwise.\qed
\end{definition}

\begin{Proposition}{\label{prop:negatively.subjective.reduct.FAEEL.world.view}}
Given a theory~$\Gamma$ and,
a total belief view~$\wv$ is FAEEL-world view of~$\Gamma$ iff
$\wv$ is a FAEEL-world view of~$\Gamma^{\underline{\wv}}$.
\end{Proposition}

\begin{theorem}\label{thm:stratified.coincide}
For epistemically tight programs, FAEEL and \mbox{G91-world} views coincide.\qed
\end{theorem}

\begin{proof}
For any belief view~$\wv$ and espistemically tight program~$\Pi$
it follows that
$\wv$ is \mbox{FAEEL-world} view of~$\Pi$ iff
$\wv$ is a \mbox{FAEEL-world} view of~$\Pi^{\underline{\wv}}$
(Proposition~\ref{prop:negatively.subjective.reduct.FAEEL.world.view}).
Furthermore, 
it is easy to see that, since $\Pi$ is espistemically tight,
$\Pi^{\underline{\wv}}$ is also espistemically tight.
Moreover, for every $r \in \Pi^{\underline{\wv}}$ and every $L \in \Bodym(r)$,
we can check that $L \in \Bodyp(r)$.
That is, there are no subjective literals in the scope of negation and, thus,
$\Pi^{\underline{\wv}}$ being espistemically tight
implies that
$\Pi^{\underline{\wv}}$ is espistemically stratified.
Then, the result follows directly from
Corollary~\ref{cor:stratified.coincide}.
\end{proof}

It is also worth to mention that, as observed in~\cite{CabalarFF2019faeel}, it is possible to obtain Moore's Autoepistemic Logic from FAEEL simply by adding the exclude middle axiom  $p \vee \neg p$ for every atom in the signature.
Note that, if the original program was stratified, augmenting it with these formulas does not change this property. As a result we obtain the following corollary:

\begin{corollary}\label{cor:stratified.Moore}
Any epistemically stratified program~$\Pi$ has at most one Moore's autoepistemic extension.\qed
\end{corollary}

A similar result was originally proved in~\cite[Theorem~4]{DBLP:conf/aaai/Gelfond87}.
Note that the class of programs considered stratified by us is slightly broader than the one used in~\cite{DBLP:conf/aaai/Gelfond87}: constraints are allowed in every strata and no distinction is made between positive and negative objective literals.
The price to pay is that Corollary~\ref{cor:stratified.Moore} does not ensure the existence of an extension.

\review{\section{Related work}\label{sec:related}}{rev3.1}

As mentioned in the introduction, the search for a ``satisfactory'' semantics for epistemic logic programs has leave us with a variety of semantics~\cite{Gelfond91,WangZ05,Truszczynski11,Gelfond11,CerroHS15,Kahl15,ShenE17,CabalarFF2019faeel}.
Among this
Epistemic Equilibrium Logic (EEL;~\citeNP{CerroHS15}) is very similar to Founded EEL in the sense that it is also defined as a combination of Equilibrium Logic and the modal logic S5.
There are some slight differences though, and as the name suggest Founded EEL satisfies the founded property defined~\cite{CabalarFF2019faeel} while EEL does not.
In fact, EEL can be characterised by a particular class of belief views that we call here~\emph{simple}:

\begin{definition}\label{def:EEL.world.view}
We say that a belief view $\wv$ is \emph{simple} iff for any $\tuple{H,T} \in \wv$ and $\tuple{H',T} \in \wv$, we have $H = H'$.
A a total belief view~$\wv$ is called an \emph{EEL-world view} of a theory~$\Gamma$
iff $\wv$ is a belief model of~$\Gamma$ and there is no simple belief model $\wv'$ of~$\Gamma$ satisfying $\wv = (\wv')^t$ and $H \subset T$ for some $\tuple{H,T} \in \wv'$.\qed
\end{definition}

It is easy to see that Definition~\ref{def:EEL.world.view} is just a rephrasing of epistemic equilibrium models as defined in~\cite{CerroHS15} by using the notation of this paper.

\begin{Proposition}{\label{prop:eel.char}}
A a total belief view~$\wv$ is an \emph{EEL-model} of a theory~$\Gamma$
iff $\wv$ is a belief model of~$\Gamma$ and there is no simple belief model $\wv'$ of $\Gamma$
such that $\wv' \prec \wv$.
\end{Proposition}

\begin{theorem}\label{thm:EEL}
Every FAEEL and \mbox{FEEL-world} view of any theory~$\Gamma$ is also an \mbox{EEL-world view.\qed}
\end{theorem}

\begin{proof}
Note that that every total belief view is simple, though there are non-total belief views that are not simple.
Then, the result follows directly from Proposition~\ref{prop:eel.char} and Theorem~\ref{prop:faeel->feel}.
\end{proof}

As illustrated by the following example, in general, the converse of Theorem~\ref{thm:EEL} does not hold.

\begin{example}\label{ex:unfounded}
Take the epistemic logic program\newprogramhide\label{theory:larger.set.of.worlds}:
\begin{gather*}
a \vee b
\hspace{2cm}
a \leftarrow \bL b
\hspace{2cm}
b \leftarrow \bL a
\tag{$\program\ref{theory:larger.set.of.worlds}$}
\end{gather*}
whose unique epistemic equilibrium model is
$\wv = \sset{ \set{a}, \set{b} }$.
Note that 
$\wv' = \sset{ \set{a , b} }$
is not an epistemic equilibrium model because $\wv'' \prec \wv'$ with $\wv'' = \sset{ \tuple{\set{a},\set{a,b}} , \tuple{\set{b},\set{a,b}}}$.
However, $\wv'$ is an EEL-model.
Note that $\wv''$ is not a simple belief view and, thus, cannot be used as a witness to show that $\wv'$ is not an EEL-model.
On the other hand, the simple belief views
$\sset{ \tuple{\set{a},\set{a,b}} }$, \
$\sset{ \tuple{\set{b},\set{a,b}}}$,
and
$\sset{ \tuple{\emptyset,\set{a,b}}}$
are not models of this program.\qed
\end{example}

Interestingly, it can be shown that EEL also satisfies epistemic splitting.

\begin{Theorem}[Epistemic splitting in EEL]{\label{thm:epistemic.splitting.eel}}
EEL satisfies the epistemic splitting property.
Furthermore, if $\wv$ is a \mbox{EEL-world} view of some program~$\Pi$ with respect to some splitting set~$U$
then $\tuple{\restr{\wv}{U},\restr{\wv}{\overline{U}}}$ is a EEL-solution of~$\Pi$ and it satisfies that $\wv = \restr{\wv}{U} \sqcup \restr{\wv}{\overline{U}}$.\qed
\end{Theorem}

The proof of Theorem~\ref{thm:epistemic.splitting.eel} is analogous to the proof of Theorem~\ref{thm:epistemic.splitting.feel} just taking into account that now we have to restrict ourselves to simple interpretations.
Note that, in general, EEL does not satisfy supra-ASP.
Example~\ref{ex:faeel.vs.feel} can be used to illustrate this statement and, in fact, the program in this example has exactly the same \mbox{EEL-world} views as \mbox{FEEL-world} views.
For this reason~\cite{CerroHS15} also included a selection of \mbox{EEL-world} views called \mbox{AEEL-world} views in a similar spirit as how FAEEL-world views are a selection of \mbox{FEEL-world} views.
However, it has been shown in~\cite{Cabalar2018nmr,CabalarFF2019splitting} that AEEL does not satisfy epistemic splitting.
Theorem~\ref{thm:epistemic.splitting.eel} sheds more light into this issue by showing that it is not the EEL logic, but the selection of \mbox{AEEL-world} views, what breaks the splitting property.
In this sense, it would be possible to define \mbox{AEEL-world} views in an alternative way as the intersection of \mbox{EEL-world} views and G91-world views and obtain yet another semantics that satisfy epistemic splitting.
Note though that this alternative semantics (EEL+G91) would not satisfy the foundedness property introduced in~\cite{CabalarFF2019faeel}.
In fact, a variation of Example~\ref{ex:unfounded}, obtained by adding the constraint $\bot \leftarrow \neg\!\bL a$ to program~\program\ref{theory:larger.set.of.worlds}, was used~\cite{CabalarFF2019faeel} to show that, among others, AEEL does not satisfy the foundedness property.
This same example can also be used to show that EEL and EEL+G91 do not satisfy it.

To summarise the state of the art, let us recall the remaining property introduced in~\cite{Cabalar2018nmr,CabalarFF2019splitting}:

\begin{property}[Supra-S5]\label{property:supraS5}
A semantics satisfies \emph{supra-S5} when for every world view $W$ of an epistemic program $\Pi$ and for every $I \in W$, $\tuple{W,I} \models \Pi$. \qed
\end{property}

\begin{table}
\begin{tabular}{  @{}p{6.5pc}  c c c c c c c c c }
&G91 & G11 & EEL & AEL & EEL+G91 & K15 & S17 & FEEL & FAEEL
\\\hline
Supra-S5& 
\checkmark & \checkmark & \checkmark & \checkmark & \checkmark & \checkmark & \checkmark & \checkmark & \checkmark
\\\hline
Supra-ASP& 
\checkmark & \checkmark &  & \checkmark & \checkmark & \checkmark & \checkmark &  & \checkmark 
\\\hline
Subjective constraint monotonicity&
\checkmark & \checkmark & \checkmark & & \checkmark & & & \checkmark & \checkmark
\\\hline
Splitting
&
\checkmark &  & \checkmark & & \checkmark & & & \checkmark & \checkmark
\\\hline
Foundness&
& & & & & & & \checkmark & \checkmark
\\\hline
\end{tabular}
	\caption{Summary of properties in different semantics.}
	\label{table:summary}
\end{table}

Table~\ref{table:summary} summarises the known results for different semantics with respect to Properties~1-\ref{property:supraS5} plus \emph{foundness}.
Recall that, intuitively, \emph{foundness} means that a semantics is free of self-supported world view.
For space reasons, we refer to~\cite{CabalarFF2019faeel} for a formal definition.
Recall also that~\cite{WangZ05,Truszczynski11} extended the semantics of G91 to arbitrary theories in different ways, but for the class of logic programs both of them agree with G91.
For this reason, we will refer to both of them just as G91 in the table.
Counterexamples for the non satified properties can be founed in~\cite{KahlL18,Cabalar2018nmr,CabalarFF2019splitting,CabalarFF2019faeel} and in this paper for the case of EEL and FEEL not satisfying \mbox{Supra-ASP}.
Proofs for the satisfied properties can be found in~\cite{Cabalar2018nmr,CabalarFF2019splitting,CabalarFF2019faeel}
and in this paper.
To complete the table, the following result shows that, despite not satisfying epistemic splitting, the semantics propsed by~\citeN{Gelfond11} does satisfy subjective constraing monotonicity.

\begin{Proposition}{\label{prop:g11.subjective.contraint.monotonicity}}
The semantics defined by~\citeN{Gelfond11} satisfies subjective constraint monotonicity.\qed
\end{Proposition}

\section{Conclusions}
\label{sec:conclusions}

We have shown that Founded Autoepistemic EEL satisfies the \emph{epistemic splitting}, a desirable property for epistemic logic programs that, among previous semantics, was known to be satisfied only by G91.
On the other hand, it is well-known that the G91 semantics suffers from \mbox{self-supported} world views, 
something that was proved to be not the case for Founded Autoepistemic EEL in~\cite{CabalarFF2019faeel}.
In this sense, Founded Autoepistemic EEL is the first semantics whose world views are not \mbox{self-supported} and that satisfies epistemic splitting.
Furthermore, we have shown that, for epistemic tight programs (those not containing cycles involving positive epistemic literals), both G91 and Founded Autoepistemic EEL coincide.
This means that Founded Autoepistemic EEL corrects the problem with self-supported world views present in G91 without introducing further variations that are unrelated to this problem.

In addition, we have introduced Founded EEL, a logic which can be considered as a combination of the stable models semantics and the modal logic~S5, and an alternative characterisation of Founded Autoepistemic EEL-world views in terms of Founded EEL and G91.
This alternative characterisation may help us to further study properties of Founded Autoepistemic EEL and, in fact, it already has been used to prove the epistemic splitting property, strengthen the relation between Founded Autoepistemic EEL and G91, and also to study the relation with the Epistemic Equilibrium Logic introduced by~\citeN{CerroHS15}.

\bibliographystyle{acmtrans}
\bibliography{refs,dblp}

\label{lastpage}

\newpage

\appendix
\section{Proofs of Results}

The proof of Theorem~\ref{thm:faeel<->feel+g91} rely on the definition of weak autoepistemic world views
which coincide with G91-world views.
Let us start by defining semi-total interpretations.
We say that an interpretation $\tuple{\wv,H,T}$ is \emph{belief-total} iff 
$\wv$ is total.
It is easy to see that, every total interpretation is belief-total but not vice-versa.

\begin{definition}
A total interpretation $\cI$ is said to be a \emph{weak \kdequilibrium model} of some theory~$\Gamma$ iff $\cI$ is a model of $\Gamma$ and there is no other belief-total model $\cI'$ of $\Gamma$ such that $\cI' \prec \cI$.\qed
\end{definition}

Note that every \kdequilibrium model is also a weak \kdequilibrium model, but not vice-versa.
For instance, $\kdint{\set{a}}{\sset{\set{a}}}$ is a weak \kdequilibrium model of $\set{a \leftarrow \bL a}$
but not a \kdequilibrium model.
By $\WMEQ[\Gamma]$ we denote the set of weak equilibrium belief models of $\Gamma$.
As in FAEEL, we impose a fixpoint condition to minimise the agent's knowledge as follows.

\begin{definition}\label{def:weal-au-eqmodel}
A total belief view~$\wv$ is called a \emph{weak autoepistemic world view} of~$\Gamma$
iff
\begin{IEEEeqnarray*}{c+x*}
\wv \ \ = \ \ \setm{ T }{ \kdint{T}{\wv } \in \WMEQ[ \Gamma ] }
&\qed
\end{IEEEeqnarray*}
\end{definition}

\begin{lemma}\label{lem:semi-total.reduct}
Let $\Gamma$ be a formula and $\cI = \tuple{\wv,H,T}$ be a semi-total interpretation.
Then, $\cI$ is a model of $\varphi$ iff $\cI$ is a model of $\varphi^\wv$.\qed
\end{lemma}

\begin{proof}
Assume that $\varphi = \bL \psi$.
Then, we have that
$\tuple{\wv,H',T'} \models \varphi$ if and only if
$\tuple{\wv,T'',T''} \models \psi$ for every $\tuple{T'',T''} \in \wv$
iff
$\wv$ is a \sfmodel of $\varphi$
iff $\varphi^\wv = \top$
iff $\cI \models \varphi^\wv$.
Then, by induction in the structure of~$\varphi$,
we get that
$\cI \models \varphi$
iff
$\cI \models \varphi^\wv$.
Finally, we have that $\cI$ is a model of $\varphi$
iff $\cI \models \varphi$ and $\tuple{\wv,T',T'} \models \varphi$ for every $\tuple{T',T'} \in \wv$
iff $\cI \models \varphi^\wv$ and $\tuple{\wv,T',T'} \models \varphi^\wv$ for every $\tuple{T',T'} \in \wv$
iff  $\cI$ is a model of $\varphi^\wv$.
\end{proof}

\begin{lemma}\label{lem:ht.correspondence}
Let $\Gamma$ be a propositional theory and $\cI = \tuple{\wv,H,T}$ be some interpretation.
Then, $\cI$ is a model of $\Gamma$ iff $\tuple{H',T'}$ is a
\htmodel of $\Gamma$ 
for every \htinterpretation $\tuple{H',T'} \in \wv \cup \set{ \tuple{H,T} }$.\qed
\end{lemma}

\begin{proof}
By definition, we have that
$\cI$ is a model of $\Gamma$
iff
$\cI$ is a model of $\varphi$
for all $\varphi \in \Gamma$.
Furthermore,
$\cI$ is a model of $\varphi$
iff $\cI \models \varphi$
and
$\tuple{\wv,H',T'} \models \varphi$ for every $\tuple{H',T'} \in \wv$
iff
$\tuple{\wv,H',T'} \models \varphi$ for every \htinterpretation $\tuple{H',T'} \in \wv \cup \set{ \tuple{H,T} }$.
Finally, since $\varphi$ is a propositional formula, we have that
$\tuple{\wv,H',T'} \models \varphi$
iff
$\tuple{H',T'} \models \varphi$.
Hence,
$\cI$ is a model of $\Gamma$ iff $\tuple{H',T'}$ is a
\htmodel of $\Gamma$ 
for every \htinterpretation $\tuple{H',T'} \in \wv \cup \set{ \tuple{H,T} }$.\qed
\end{proof}

\begin{theorem}\label{thm:g91.weak}
Given a theory~$\Gamma$ and some belief view~$\wv$,
we have that $\wv$ is a weak autoepistemic world view of~$\Gamma$ iff $\wv$ is a G91-world view of $\Gamma$.\qed
\end{theorem}

\begin{proof}
From Lemmas~\ref{lem:semi-total.reduct} and~\ref{lem:ht.correspondence},
we have that
$\tuple{\wv,T} \in \WEQ[ \Gamma ]$ iff $\tuple{\wv,T} \in \WEQ[ \Gamma^\wv ]$
iff $\wv \cup \set{T} \subseteq \SM[ \Gamma^\wv ]$
and, therefore, Definition~\ref{def:weal-au-eqmodel} can be rewriting as
\begin{gather*}
\wv \ \ = \ \ \setm{ T }{ \wv \cup \set{T} \subseteq \SM[ \Gamma^\wv ] }
\end{gather*}
iff
$\wv = \setm{ T }{ T \in \SM[ \Gamma^\wv ] }$
iff $\wv = \SM[ \Gamma^\wv ]$
iff $\wv$ is a G91-world view of $\Gamma$.
\end{proof}

\begin{proposition}\label{prop:a.world.view->weak}
Every FAEEL-world view is also a weak autoepistemic world view.\qed
\end{proposition}

\begin{proof}
Since every \kdequilibrium model is also a weak \kdequilibrium model,
it only remains to be shown that if $\wv$ is an autoepistemic world view then
\begin{enumerate}
\item[]there not exists any propositional interpretation $T$ such that $\tuple{\wv,T}$ is a weak \kdequilibrium model of $\Gamma$ and $T \notin \wv$.
\end{enumerate}
Suppose, for the sake of contradiction, that 
there is some propositional interpretation $T$ such that $\tuple{\wv,T}$ is a weak \kdequilibrium model of $\Gamma$ and $T \notin \wv$.
Since $\wv$ is an autoepistemic world view, this implies that
$\tuple{\wv,T}$ is not a \kdequilibrium model of $\Gamma$ and, thus, that there is some non-semi-total model $\cI'=\tuple{\wv',H,T}$ of $\Gamma$ such that $\cI' \prec \cI$.
Hence, there is some $\tuple{H',T'} \in \wv'$ such that $H' \subset T'$.
Let $\cI'' = \tuple{\wv',H',T'}$.
Then, we have that $\cI'' \prec \tuple{\wv,T'}$ and, since $\wv$ is an autoepistemic world view of $\Gamma$ and $T \in \wv$, we have that $\tuple{\wv,T'}$ is a \kdequilibrium model of $\Gamma$.
These two facts together imply that $\cI''$ is not a model of $\Gamma$.
Hence, there is a formula~$\varphi\in\Gamma$ such that $\cI''$ is not a model of $\varphi$ and, thus,
there is $\tuple{H'',T''} \in \wv' \cup \set{\tuple{H',T'}}$ such that
$\tuple{\wv',H'',T''} \not\models \varphi$.
On the other hand, since
$\cI'$ is a model of $\Gamma$ ,
it follows that
$\tuple{\wv',H''',T'''} \not\models \varphi$
for every $\tuple{H''',T'''} \in \wv' \cup \set{\tuple{H,T}}$.
Hence, it follows that $H''=H'$ and $T''=T'$ and that
$\tuple{\wv',H',T'} \not\models \varphi$.
However, since $\tuple{H',T'} \in \wv'$,
this implies that
$\cI' = \tuple{\wv',H,T} \not\models \varphi$, which is a contradiction with the fact that $\cI'$ is a model of $\Gamma$.
Consequently,
$\wv$ is a weak autoepistemic world view.
\end{proof}

\begin{lemma}\label{lem:epistemic.decomposition}
Let $\Gamma$ be a theory,
$\wv$ be an FEEL-world view of~$\Gamma$ and $T \in \wv$ be a propositional interpretation.
If $\tuple{\wv,T}$ is a weak equilibrium belief model of~$\Gamma$,
then $\tuple{\wv,T}$ is also a equilibrium belief model of~$\Gamma$.
\end{lemma}

\begin{proof}
Suppose that $\tuple{\wv,T}$ is not a \kdequilibrium model of $\Gamma$.
Since $\tuple{\wv,T}$ is a weak \kdequilibrium model of $\Gamma$,
we have that $\kdint{T}{\wv}$ is a belief model of $\Gamma$
and, thus, there must be some \mbox{non-semi-total} model $\cI'=\tuple{\wv',H,T}$ of $\Gamma$ such that
$\cI' \prec \kdint{T}{\wv}$.
Hence, there is $\tuple{H',T'} \in \wv'$ with $H \subset T$.
This implies that $\wv' \prec \wv$ and, since $\wv$ is an \mbox{FEEL-world} view of~$\Gamma$,
that $\wv$ is not an epistemic model of~$\Gamma$
Hence, $\wv'$ is not a epistemic model of $\Gamma$ either.
This implies that there is some formula $\varphi \in \Gamma$ and \htinterpretation~$\tuple{H'',T''} \in \wv'$
such that $\tuple{\wv',H'',T''} \not\models \varphi$.
In its turn, this implies that $\cI'$ is not  belief model of $\Gamma$ which is a contradiction.
Consequently, $\tuple{\wv,T}$ is a equilibrium belief model of~$\Gamma$.
\end{proof}

\begin{Proofof}{\ref{thm:faeel<->feel+g91}}
First note that from Proposition~\ref{prop:faeel->feel} and~\ref{prop:a.world.view->weak},
it follows that every \mbox{FAEEL-world} view is both an FEEL-world view and a weak autoepistemic world view of~$\Gamma$.
The other way around.
Since $\wv$ is a weak autoepistemic equilibrium model of $\Gamma$,
it follows that $\tuple{\wv,T}$ is a weak \kdequilibrium model of $\Gamma$ for every $T \in \wv$
which, from Lemma~\ref{lem:epistemic.decomposition},
implies that
$\tuple{\wv,T}$ is a equilibrium belief model of $\Gamma$ for every $T \in \wv$.
Hence, it only remains to be shown that the following condition hold:
\begin{enumerate}
\item[] there is no propositional interpretation $T$ such that model $\tuple{\wv,T}$ is a equilibrium belief model of $\Gamma$ and $T \notin \wv$.
\end{enumerate}
Suppose, for the sake of contradiction, that
there is some propositional interpretation $T$ such that model $\tuple{\wv,T}$ is an equilibrium belief model of $\Gamma$ and $T \notin \wv$.
First, this implies that $\tuple{\wv,T}$ is an belief model of $\Gamma$.
Then, since $\wv$ is a weak autoepistemic equilibrium model of $\Gamma$ and $T \notin \wv$,
it follows that 
$\tuple{\wv,T}$
cannot be a weak autoepistemic equilibrium model of $\Gamma$ and, thus,
there is some semi-total belief model $\cJ$ is a model of $\Gamma$ such that $\cJ \prec \tuple{\wv,T}$.
But every semi-total belief model is also a belief model, so this is a contradiction with the fact that $\tuple{\wv,T}$ is a equilibrium belief model of $\Gamma$.
Consequently, $\wv$ must be an autoepistemic world view of $\Gamma$.
\end{Proofof}

\begin{proposition}[Free atom invariance]\label{prop:free.atom.invariance}
Given a  set of atoms $U \subseteq \at$,
a formula~$\varphi$ such that 
$\Atoms(\varphi) \subseteq U$
and
a pair of belief views $\wv$ and $\wv'$
such that 
$\restr{\wv}{U} = \restr{\wv'\!}{U}$,
it follows that $\wv \models \varphi$ iff $\wv' \models\varphi$.\qed
\end{proposition}

\begin{proof}
By definition, it follows that $\wv \models \varphi$ iff
$\kdint{H_i,T_i}{\wv} \models \varphi$ for all \mbox{$\tuple{H_i,T_i} \in \wv$}
and that
$\wv' \models \varphi$ iff
$\kdint{H_i,T_i}{\wv'} \models \varphi$ for all \mbox{$\tuple{H_i',T_i'} \in \wv'$}.
Then, if $\varphi \in \at$ is an atom, 
we have $\varphi \in U$ and, thus,
it follows that
$\wv \models \varphi$ iff
$\kdint{H_i,T_i}{\wv} \models \varphi$ for all \mbox{$\tuple{H_i,T_i} \in \wv$}
iff
$a \in H_i$ for all \mbox{$\tuple{H_i,T_i} \in \wv$}
iff
$a \in H_i \cap U$ for all \mbox{$\tuple{H_i,T_i} \in \wv$}
iff
$a \in H_i'$ for all \mbox{$\tuple{H_i',T_i'} \in \restr{\wv}{U} = \restr{\wv'\!}{U}$}
iff
$\wv' \models \varphi$.
The rest of the proof follows by induction in the structure of the formula.
\end{proof}

Let us now extended the operation~$\sqcup$ to possibly non-total belief views as follows:
$$\wv_b \sqcup \wv_t \ \  = \ \ \setm{\tuple{H_b \cup H_t,T_b \cup T_t }}{ \tuple{H_b,T_b} \in \wv_b  \text{ and } \tuple{H_t,T_t} \in \wv_t  }$$

\begin{lemma}\label{lem:aux1:prop:splitting.aux}
Let $U \subseteq \at$ be some set of atoms
and let $\wv$ and $\wv_1$ be two belief views such that $\wv_1 \preceq \restr{\wv}{U}$.
Then, $\wv_1 \sqcup \restr{\wv}{\,\overline{U}} \preceq \wv$.
\end{lemma}

\begin{proof}
Let $\wv' = \wv_1 \sqcup \restr{\wv}{\,\overline{U}}$
and suppose, for the sake of contradiction, $\wv' \not\preceq \wv$.
Then
\begin{enumerate}[ leftmargin=18pt , label=(\roman*) ]
\item there is $\tuple{H,T} \in \wv$ such that
$H' \not\subseteq H$
for all $\tuple{H',T} \in \wv'$, or

\item there is $\tuple{H',T} \in \wv'$
such that
$H'  \not\subseteq H$
for all $\tuple{H,T} \in \wv$.
\end{enumerate}
If the former,
$\tuple{H,T} \in \wv$
implies that
$\tuple{H \cap U,T \cap U} \in \restr{\wv}{U}$
and
$\tuple{H \cap \overline{U},T \cap \overline{U}} \in \restr{\wv}{\overline{U}}$.
Furthermore,
$\tuple{H \cap U,T \cap U} \in \restr{\wv}{U}$
plus $\wv_1 \preceq \restr{\wv}{U}$
imply that
there is
$\tuple{H'',T\cap U} \in \wv_1$ such that $H'' \subseteq H \cap U$.
Let $H' = H'' \cup (H \cap \overline{U})$.
Then, $\tuple{H',T} \in \wv_1 \cup \restr{\wv}{\overline{U}}= \wv'$
which is a contradiction.
\\[5pt]
If the latter,
$\tuple{H',T} \in \wv'$
implies that there are $\tuple{H_1',T_1} \in \wv_1$ and $\tuple{H_2,T_2} \in \restr{\wv}{\overline{U}}$
such that $H' = H_1' \cup H_2$ and $T = T_1 \cup T_2$.
Furthermore,
$\tuple{H_1',T_1} \in \wv_1$ 
plus
$\wv_1 \preceq \restr{\wv}{U}$
imply that
there is
$\tuple{H_1,T_1} \in \restr{\wv}{U}$ such that 
$H_1' \subseteq H_1$.
Let $H = H_1  \cup H_2$.
Then, $\tuple{H,T} \in \wv$
and $H' = H_1' \cup H_2 \subseteq H_1 \cup H_2 = H$
which is a contradiction.
\end{proof}

\begin{lemma}\label{lem:aux2:prop:splitting.aux}
Let $U \subseteq \at$ be some set of atoms
and let $\wv$ and $\wv_1$ be two belief views such that $\wv_1 \prec \restr{\wv}{U}$.
Then, $\wv_1 \sqcup \restr{\wv}{\,\overline{U}} \prec \wv$.
\end{lemma}

\begin{proof}
Let $\wv' = \wv_1 \sqcup \restr{\wv}{\,\overline{U}}$.
By definition,
$\wv_1 \prec \restr{\wv}{U}$
implies
$\wv_1 \preceq \restr{\wv}{U}$
and, from Lemma~\ref{lem:aux1:prop:splitting.aux},
this implies
$\wv' \preceq \wv$.
Furthermore, since $\wv_1 \prec \restr{\wv}{U}$,
it follows that
there is $\tuple{H_1',T_1} \in \wv_1$ and $\tuple{H_1,T_1} \in \restr{\wv}{U}$
such that $H_1' \subset H_1$.
In addition,
$\tuple{H_1,T_1} \in \restr{\wv}{U}$
implies that there is $\tuple{H,T} \in \wv$
such that $H_1 = H \cap U$ and $T_1 = T \cap U$.
Let $H' = H_1' \cup (H \cap \overline{U})$.
Then, $\tuple{H',T} \in \wv_1 \sqcup \restr{\wv}{\overline{U}} = \wv'$ and we can see that $H' \subset H$.
Consequently, $\wv' \prec \wv$.
\end{proof}

\begin{lemma}\label{lem:preceq.project}
Let $\wv$ and $\wv'$ be two belief views such that $\wv \preceq \wv'$ and $U \subseteq \at$ be a set of atoms.
Then, $\restr{\wv}{U} \preceq \restr{\wv'\!}{U}$.
\end{lemma}

\begin{proof}
Suppose, for the sake of contradiction, that 
$\restr{\wv}{U} \not\preceq \restr{\wv'\!}{U}$.
Then,
one of the following hold:
\begin{enumerate}[ leftmargin=18pt , label=(\roman*) ]
\item there is $\tuple{H',T} \in \restr{\wv'\!}{U}$
such that $H \subseteq H'$
for every $\tuple{H,T} \in \restr{\wv}{U}$,

\item there is $\tuple{H,T} \in \restr{\wv}{U}$ such that
$H \subseteq H'$
for every $\tuple{H',T} \in \restr{\wv'\!}{U}$,
\end{enumerate}
Assume the former.
Then, there is $\tuple{H_1',T_1} \in \wv'$ such that $H' = H_1' \cap U$ and $T = T_1 \cap U$.
Furthermore, since $\wv \preceq \wv'$, this implies that
there is $\tuple{H_1,T_1} \in \wv$ such that $H_1 \subseteq H_1'$.
In its turn, this implies that $\tuple{H,T} \in \restr{\wv}{U}$
with $H = H_1 \cap U$ and, it is easy to see that $H = H_1 \cap U \subseteq H_1' \cap U = H'$.
which is a contradiction with the assumption.
\\[10pt]
Otherwise, there is
$\tuple{H_1,T_1} \in \wv'$ such that $H = H_1 \cap U$ and $T = T_1 \cap U$.
Furthermore, since $\wv \preceq \wv'$, this implies that
there is $\tuple{H_1',T_1} \in \wv'$ such that $H_1 \subseteq H_1'$.
In its turn, this implies that $\tuple{H',T} \in \restr{\wv'\!}{U}$
with $H' = H_1' \cap U$ and, it is easy to see that $H = H_1 \cap U \subseteq H_1' \cap U = H'$.
which is a contradiction with the assumption.
\end{proof}

\begin{lemma}\label{lem:prec.project}
Let $\wv$ and $\wv'$ be two belief views such that $\wv \prec \wv'$ and $U \subseteq \at$ be a set of atoms.
Then, $\restr{\wv}{U} \prec \restr{\wv'\!}{U}$ or $\restr{\wv}{\overline{U}} \prec \restr{\wv'\!}{\overline{U}}$.
\end{lemma}

\begin{proof}
From Lemma~\ref{lem:preceq.project}, it follows that
$\restr{\wv}{U} \preceq \restr{\wv'\!}{U}$ or $\restr{\wv}{\overline{U}} \preceq \restr{\wv'\!}{\overline{U}}$
holds.
Then, it is easy to see that
$\restr{\wv}{U} = \restr{\wv'\!}{U}$ or $\restr{\wv}{\overline{U}} = \restr{\wv'\!}{\overline{U}}$
implies
$\wv = \wv'$.
Hence, the result must hold.
\end{proof}

\begin{Proofof}{\ref{prop:splitting.aux}}
Assume first that $\wv$ is an equilibrium epistemic model of $\Pi$.
Then, $\wv$ is a total epistemic model of $\Pi$
and, thus, also of $\Pi_1$ and $\Pi_2$.
From Proposition~\ref{prop:free.atom.invariance}, this implies that
$\restr{\wv}{U}$ is a total epistemic model of $\Pi_1$
and
$\restr{\wv}{\overline{U}}$ is a total epistemic model of $\Pi_2$.
Suppose now, for the sake of contradiction, that
$\restr{\wv}{U}$ is not an equilibrium epistemic model of $\Pi_1$.
Then, there is some non-total epistemic model $\wv_1$ of $\Pi_1$ such that $\wv_1 \prec \restr{\wv}{U}$
and, from Proposition~\ref{prop:free.atom.invariance} again, it follows that $\wv' = \wv_1 \sqcup \restr{\wv}{\overline{U}}$
is a non-total epistemic model of $\Pi$.
Furthermore, it can be checked that $\wv' \prec \wv$ (Lemma~\ref{lem:aux2:prop:splitting.aux}).
This is a contradiction with the assumption and, thus,
$\restr{\wv}{U}$ must be an equilibrium epistemic model of $\Pi_1$.
The proof that
$\restr{\wv}{\overline{U}}$ must be an equilibrium epistemic model of $\Pi_2$ is analogous.
\\[10pt]
The other way around.
Assume that
$\restr{\wv}{U}$ is an equilibrium epistemic model of $\Pi_1$
and
$\restr{\wv}{\overline{U}}$ is an equilibrium epistemic model of $\Pi_2$
Again, from Proposition~\ref{prop:free.atom.invariance}, this implies that $\wv$ is epistemic model of $\Pi$.
Furthermore, if we suppose that $\wv$ is not equilibrium epistemic model of $\Pi$,
then there must be some $\wv' \prec \wv$ such that $\wv'$ is a non-total epistemic model of $\Pi$ and this implies that
$\restr{\wv'\!}{U}$ is epistemic model of $\Pi_1$
and
$\restr{\wv'\!}{\overline{U}}$ is epistemic model of $\Pi_2$
and that
either
\mbox{$\restr{\wv'\!}{U} \prec \restr{\wv}{U}$}
or
\mbox{$\restr{\wv'\!}{\overline{U}} \prec \restr{\wv}{\overline{U}}$}
holds Lemma~\ref{lem:prec.project}.
This is a contradiction with the assumption, which implies that
$\wv$ must be an equilibrium epistemic model of~$\Pi$.
\end{Proofof}

\begin{lemma}\label{lem:splitting.aux1}
Let $\Gamma$ be some theory, $\wv$ be some total belief view and $\varphi$ be some formula such that $\wv \models \bL \varphi$ (resp.  $\wv \not\models \bL \varphi$).
Let $\Gamma'$ be the result of replacing some occurrence of $\bL \psi$ by $\top$ (resp.~$\bot)$.
Then, $\wv \models \Gamma$ iff $\wv \models \Gamma'$.\qed
\end{lemma}

\begin{proof}
It is enough to show that $\wv \models \varphi$ iff $\wv \models \varphi'$ for any formula~$\varphi \in \Gamma$.
Note that, if $\varphi$ is a propositional formula, then $\varphi = \varphi'$ and the result holds.
Furthermore, if $\varphi = \bL \psi$, we have $\wv  \models \bL \psi$
and $\wv  \models \top$
(resp.  $\wv  \not\models \bL \psi$ and $\wv  \not\models \bot$), so the result also holds.
Then, the rest of the proof follows by structural induction.
\end{proof}

\begin{lemma}\label{lem:splitting.aux2}
Let $\Pi$ be some logic program, $\wv$ be some (possibly non-total) belief interpretation and $a \in \at$.
Let $\Pi'$ be the result of replacing some occurrence of $\bL a$ in the body of some rule by $\top$
if $\wv \models \bL a$ and by $\bot$ otherwise.
Then, $\wv \models \Pi'$ implies that $\wv \models \Pi$.\qed
\end{lemma}

\begin{proof}
It is enough to prove that any rule $r \in \Pi$
satisfies that
$\wv \models r'$ implies $\wv \models r$
where $r'$ is the result of replacing some occurrence of $a$ in its body by $\top$.
Then,
we assume that $\wv \models r'$
and we will show that
$\kdint{H_i,T_i}{\wv} \models r$
for any $\tuple{H_i,T_i} \in \wv$.
Note that $\wv \models r'$ implies
$\kdint{H_i,T_i}{\wv} \models r'$
and, thus, that one of the following conditions hold:
\begin{itemize}[ leftmargin=18pt  ]
\item $\Head(r') \cap H_i \neq \emptyset$, or
\item $\top \in \Bodyn(r')$ ,or
\item $\bot \in \Bodyp(r')$, or
\item $\Atoms(\Bodyrp(r')) \not\subseteq T_i$, or
\item $\Atoms(\Bodymp(r')) \not\subseteq T_j$ for some $\tuple{H_j,T_j} \in \wv$, or
\item $\Atoms(\Bodyrn(r')) \cap T_i \neq \emptyset$ or,
\item $\Atoms(\Bodymn(r')) \cap T_j \neq \emptyset$ for some $\tuple{H_j,T_j} \in \wv$, or
\item $\Head(r') \cap T_i \neq \emptyset$ and one of the following hold:
 \begin{itemize}[ leftmargin=18pt  ]
 \item $\Atoms(\Bodyrp(r')) \cap H_i \not\subseteq I^h(w)$
 \item $\Atoms(\Bodymp(r')) \not\subseteq H_j$ for some $\tuple{H_j,T_j} \in \wv$,
 % \item $\Bodyn(r') \cap \Bodyr(r') \cap I^t(w) \neq \emptyset$
 % \item $\Bodyn(r') \cap \Bodym(r') \cap I^h(w') \neq \emptyset$ for some $w' \in W$
 \end{itemize}
\end{itemize}
Since $\Head(r) = \Head(r')$ and $\Atoms(\Body^x_y)(r') \subseteq \Atoms(\Body^x_y)(r)$ with $x \in \set{+,-}$ and $y \in \set{\mathit{obj},\mathit{sub}}$, we immediately get that one of the following condition hold:
\begin{itemize}[ leftmargin=18pt  ]
\item $\Head(r) \cap H_i \neq \emptyset$, or
\item $\bL a \in \Bodyn(r)$ and $\wv \models \bL a$, or
\item $\bL a \in \Bodyp(r)$ and $\wv \not\models \bL a$, or
\item $\Atoms(\Bodyrp(r)) \not\subseteq T_i$, or
\item $\Atoms(\Bodymp(r)) \not\subseteq T_j$ for some $\tuple{H_j,T_j} \in \wv$, or
\item $\Atoms(\Bodyrn(r)) \cap T_i \neq \emptyset$ or,
\item $\Atoms(\Bodymn(r)) \cap T_j \neq \emptyset$ for some $\tuple{H_j,T_j} \in \wv$, or
\item $\Head(r) \cap T_i \neq \emptyset$ and one of the following hold:
 \begin{itemize}[ leftmargin=18pt  ]
 \item $\Atoms(\Bodyrp(r)) \cap H_i \not\subseteq I^h(w)$
 \item $\Atoms(\Bodymp(r)) \not\subseteq H_j$ for some $\tuple{H_j,T_j} \in \wv$,
 % \item $\Bodyn(r') \cap \Bodyr(r') \cap I^t(w) \neq \emptyset$
 % \item $\Bodyn(r') \cap \Bodym(r') \cap I^h(w') \neq \emptyset$ for some $w' \in W$
 \end{itemize}
\end{itemize}
and, thus, we get that $\kdint{H_i,T_i}{\wv} \models r$.
Note that, if 
$\wv \models \bL a$
it follows 
$\wv^t \models \bL a$,
and, thus,
$\bL a \in \Bodyn(r)$ implies
$\kdint{H_i,T_i}{\wv} \models r$.
Otherwise,
$\wv \not\models \bL a$
and, thus,
$\bL a \in \Bodyp(r)$ implies
$\kdint{H_i,T_i}{\wv} \models r$.
Finally, since $\tuple{H_i,T_i} \in \wv$ was arbitrarily chosen, it follows that $\wv \models r$.
\end{proof}

\begin{Proofof}{\ref{prop:splitting.aux2}}
First note that
$B_U(\Pi) \cup E_U(\Pi,\wv)$
is the result of replacing
the occurrences of $\bL a$ in $T_U(\Pi)$ by $\top$ if $\wv \models \bL a$ and by $\bot$ otherwise.
Furthermore, since $\wv$ is total, this implies that $\wv$ is an epistemic model of~$\Pi$
if and only if $\wv$ is an epistemic model of~$B_U(\Pi) \cup E_U(\Pi,\wv)$ (Lemma~\ref{lem:splitting.aux1}).
\\[10pt]
Assume that
$\wv$ is an \mbox{FEEL-world} view of~$\Pi$
and suppose, for the sake of contradiction, that
$\wv$ is not an \mbox{FEEL-world} view of~$B_U(\Pi) \cup E_U(\Pi,\wv)$.
Then, there is some non-total epistemic model $\wv'$ of~$B_U(\Pi) \cup E_U(\Pi,\wv)$ such that $\wv' \prec \wv$, but this implies that $\wv'$ is also an epistemic model of $\Pi$  (Lemma~\ref{lem:splitting.aux2}) and, thus, that
$\wv$ is not an \mbox{FEEL-world} view of~$\Pi$.
This is a contradiction and, consequently,
$\wv$ must be an \mbox{FEEL-world} view of~$B_U(\Pi) \cup E_U(\Pi,\wv)$.
\\[10pt]
Assume now that
$\wv$ is an \mbox{FEEL-world} view of~$B_U(\Pi) \cup E_U(\Pi,\wv)$
and suppose that
$\wv$ is not an \mbox{FEEL-world} view of~$\Pi$.
Then, there is some non-total epistemic model $\wv'$ of~$\Pi$ such that $\wv' \prec \wv$ and that $\wv'$ is not an epistemic model of~$B_U(\Pi) \cup E_U(\Pi,\wv)$.
This implies that there is some rule $r \in T_U(\Pi)$ such that $\wv' \not\models r_U^\wv$
but $\wv' \models r$.
Hence, there is a subjective literal of the form $\bL a$ in the body of $r$ such that $a \in U$ and $\wv' \not\models \bL a$ but $\wv \models \bL a$.
Furthermore, from Proposition~\ref{prop:free.atom.invariance} and the fact that
$\wv$ and $\wv'$ are epistemic models of~$\Pi$,
it also follows that $\restr{\wv}{U}$ and $\restr{\wv'\!}{U}$ are epistemic models of~$B_U(\Pi)$.
Note that, since $a \in U$, we also have $\restr{\wv}{U} \prec \restr{\wv'\!}{U}$ (Lemma~\ref{lem:preceq.project})
and, thus, 
$\restr{\wv}{U}$ cannot be an \mbox{FEEL-world} views of~$B_U(\Pi)$.
From Proposition~\ref{prop:splitting.aux},
this implies that
$\wv$ cannot be an \mbox{FEEL-world} views of~$B_U(\Pi) \cup E_U(\Pi,\wv)$,
which is a contradiction with the assumptions.
Consequently, $\wv$ must be an \mbox{FEEL-world} view of~$\Pi$.
\end{Proofof}

\begin{proposition}\label{prop:negation.there}
Let $\varphi$ be a formula 
and let $\tuple{\wv,H,T}$ be an interpretation.
Then, $\tuple{\wv,H,T}\models \neg\varphi$ iff $\tuple{\wv^t,T,T} \not\models \varphi$ iff $\tuple{\wv^t,T,T} \models \neg\varphi$.\qed
\end{proposition}

\begin{proof}
By definition,
we have that
$\tuple{\wv^t,T,T} \models \neg\varphi$
iff
$\tuple{\wv^t,T,T} \models \varphi\to\bot$
iff
$\tuple{\wv^t,T,T} \not\models \varphi$.
Furthermore, by definition, we have that
$\tuple{\wv,H,T}\models \neg\varphi$
iff both
\mbox{$\tuple{\wv,H,T} \not\models \varphi$}
and
$\tuple{\wv,T,T} \not\models \varphi$.
Finally, since Proposition~1 in~\cite{CabalarFF2019faeel},
we have that
$\tuple{\wv,H,T}\models \varphi$
implies
$\tuple{\wv,T,T}\models\varphi$
we get that
$\tuple{\wv,H,T}\models \neg\varphi$
iff
$\tuple{\wv,T,T} \not\models \varphi$.
\end{proof}

\begin{proposition}\label{prop:negatively.subjective.reduct.belief.model}
Given a theory~$\Gamma$ and,
a belief interpretation~$\tuple{\wv,H,T}$ is a belief model of~$\Gamma$ iff
$\tuple{\wv,H,T}$ is a belief model of~$\Gamma^{\underline{\wv^t}}$.
\end{proposition}

\begin{proof}
Pick a formula of the form $\neg\bK\varphi$.
Then, from Proposition~\ref{prop:negation.there},
it follows
$\tuple{\wv,H,T} \models \neg\!\bK\varphi$
iff
$\tuple{\wv^t,T,T} \models \neg\!\bK\varphi$
iff
$\tuple{\wv^t,T,T} \not\models \bK\varphi$
iff
$\wv^t \not\models \bK\varphi$
iff $(\neg\!\bK\varphi)^{\wv^t} = \top$
iff
$\tuple{\wv,H,T} \models (\neg\!\bK\varphi)^{\wv^t}$.
The rest of the proof follows by induction.
\end{proof}

\begin{Proofof}{\ref{prop:negatively.subjective.reduct.FAEEL.world.view}}
From~Proposition~\ref{prop:negatively.subjective.reduct.belief.model},
it follows that $\tuple{\wv,T}$ is a equilibrium belief model of~$\Gamma$ iff
$\tuple{\wv,T}$ is a equilibrium belief model of~$\Gamma^{\underline{\wv}}$.
Hence, it immediately follows that
$\wv$ is \mbox{FAEEL-world} view of~$\Gamma$ iff
$\wv$ is a FAEEL-world view of~$\Gamma^{\underline{\wv}}$.
\end{Proofof}

\begin{Proofof}{\ref{prop:eel.char}}
Let $\wv$ be a total belief model of~$\Gamma$ and $\wv'$ be a simple belief model of~$\Gamma$
Then, it is enough to show that $\wv' \prec \wv$ iff $\wv = (\wv')^t$ and there is some $\tuple{H,T} \in \wv'$ satisfying $H \subset T$.
Assume first that
$\wv = (\wv')^t$ and there is some $\tuple{H,T} \in \wv$ satisfying $H \subset T$.
Then, $\wv = (\wv')^t$ implies
\begin{enumerate}[ leftmargin=18pt , label=(\roman*) ]
\item for every $\tuple{T,T} \in \wv$,
there is some $\tuple{H,T} \in \wv'$
with $H \subseteq T$, and

\item for every $\tuple{H,T} \in \wv'$,
there is some $\tuple{T,T} \in \wv$
with $H \subseteq T$.
\end{enumerate}
Hence, $\wv' \preceq \wv$.
Furthermore, since $\wv$ is total and there is some $\tuple{H,T} \in \wv$ satisfying $H \subset T$,
it follows that $\wv' \neq \wv$ and, thus, $\wv' \prec \wv$.
The other way around. 
$\wv' \prec \wv$
implies $\wv' \preceq \wv$ and, thus, that both (i) and (ii) hold.
In its turn, this implies $\wv = (\wv')^t$.
Furthermore, since $\wv \neq \wv'$, either there is some
$\tuple{T,T} \in (\wv \setminus \wv')$
or some
$\tuple{H,T} \in (\wv' \setminus \wv)$.
If the former, from (i), there is
$\tuple{H,T} \in \wv'$ with $H \subseteq T$.
Moreover, since $\tuple{T,T} \notin \wv'$, it immediately follows that $H \subset T$.
If the latter, there is some $\tuple{T,T} \in \wv$
with $H \subseteq T$
and, since $\tuple{H,T} \notin \wv$, that $H \subset T$.
\end{Proofof}

\begin{Proofof}{\ref{prop:g11.subjective.contraint.monotonicity}}
Suppose not.
Then there is a some program $\Pi$, belief view~$\wv$ and some subjective constraint~$r$ such that the following equivalence does not hold:
\begin{itemize}[leftmargin=10pt]
\item[] $\wv$ is a G11-world view of $\Pi \cup \set{r}$
iff $\wv$ is a G11-world view of $\Pi$ and $\wv \models r$.
\end{itemize}
Assume first that
$\wv$ is a G11-world view of $\Pi \cup \set{r}$.
Then, $\wv$ is the non-empty set of all stable models of 
$(\Pi \cup \set{r})^\wv = \Pi^\wv \cup \set{r}^\wv$,
that is, $\wv = \SM[\Pi^\wv \cup \set{r}^\wv] \neq \emptyset$.
Furthermore, we can see that
$\wv = \SM[\Pi^\wv \cup \set{r}^\wv] \subseteq \SM[\Pi^\wv]$
implies that every $M \in \wv$ satisfies
$M \models r^\wv$.
Let $r$ be of the form:
\begin{gather}
\bot \leftarrow \bK a_1 \wedge \dotsc \wedge \bK a_n \wedge \neg \bK b_1 \wedge \dotsc \wedge \neg \bK b_m
	\label{eq:1:prop:g11.subjective.contraint.monotonicity}
\end{gather}
and suppose, for the sake of contradiction, that 
$r^\wv$ is of the form:
\begin{gather}
\bot \leftarrow a_1 \wedge \dotsc \wedge a_n
	\label{eq:2:prop:g11.subjective.contraint.monotonicity}
\end{gather}
Then, for all $1 \leq i \leq n$,
we get
$\wv \models \bK a_i$
and, thus, every $M \in \wv$
satisfies
$M \models a_i$.
However, this implies every $M \in \wv$ does not statisfy $r^\wv$,
which is a contradiction with the fact that $M$ is a stable model of
$\Pi^\wv \cup \set{r}^\wv$.
Hence, it must be that $\wv\models \bK a_i$ for some $1 \leq i \leq n$ or
$\wv\not\models \bK a_j$ for some $1 \leq j \leq m$
and, thus, $\wv \models r$.
Furthermore, this implies that
$\set{r}^\wv = \emptyset$ and, thus, that
$\Pi^\wv \cup \set{r}^\wv = \Pi^\wv$.
Consequently, we get that
$\wv = \SM[\Pi^\wv \cup \set{r}^\wv] = \SM[\Pi^\wv]$ 
\\[10pt]
The other way arround.
Let $\wv$ be a G11-world view of $\Pi$.
Then, since $\wv \models r$,
it follows that $\set{r^\wv} = \emptyset$
and, thus
$\SM[\Pi^\wv \cup \set{r}^\wv] = \SM[\Pi^\wv] = \wv$.
\end{Proofof}

\end{document}